\documentclass[11pt]{article}
\usepackage{fullpage}           
\usepackage{amsfonts}           
\usepackage{amsthm}             
\usepackage{amsmath}
\usepackage{amssymb}
\usepackage{xspace}             
\usepackage{graphicx}            
\usepackage{adjustbox}          

\usepackage{multicol}            
\usepackage{url}        
\usepackage{thmtools}
\usepackage{thm-restate}
\usepackage{enumerate}  
\usepackage{subcaption}
\usepackage[usenames]{xcolor} 

\usepackage{hyperref}
\usepackage{mathabx}

\usepackage{enumitem}
\usepackage{algorithm}
\usepackage[noend]{algpseudocode}

\usepackage{stackengine}
\stackMath
\newcommand\tsup[2][2]{%
 \def\useanchorwidth{T}%
  \ifnum#1>1%
    \stackon[-1.3ex]{\tsup[\numexpr#1-1\relax]{#2}}{\mathchar"307E}%
  \else%
    \stackon[-1ex]{#2}{\mathchar"307E}%
  \fi%
}

\usepackage{nccmath}
\newtheorem{theorem}{Theorem}
\newtheorem{lemma}[theorem]{Lemma}

\newtheorem{definition}[theorem]{Definition}
\newtheorem{obs}[theorem]{Observation}

\newtheorem{claim}[theorem]{Claim}

\newtheorem{problem}[theorem]{Problem}

\newcommand{\tOh}{\widetilde{O}}

\newcommand{\wmax}{w_{\textup{max}}}
\newcommand{\pmax}{p_{\textup{max}}}
\DeclareMathOperator{\supp}{supp}
\newcommand{\cW}{{\cal W}}

\newcounter{sideremark}

\usepackage{setspace}

\title{Knapsack with Small Items in Near-Quadratic Time}
\date{}

\author{
  Karl Bringmann\thanks{Saarland University and Max-Planck-Institute for Informatics, Saarland Informatics Campus, Saarbr\"ucken, Germany.
   \texttt{bringmann@cs.uni-saarland.de}. This work is part of the project TIPEA that has received funding from the European Research Council (ERC) under the European Unions Horizon 2020 research and innovation programme (grant agreement No. 850979).
  }
}

\begin{document}

\maketitle

\begin{abstract}

The Knapsack problem is one of the most fundamental NP-complete problems at the intersection of computer science, optimization, and operations research.
A recent line of research worked towards understanding the complexity of pseudopolynomial-time algorithms for Knapsack parameterized by the maximum item weight $\wmax$ and the number of items~$n$. 
A conditional lower bound rules out that Knapsack can be solved in time $O((n+\wmax)^{2-\delta})$ for any $\delta > 0$ [Cygan, Mucha, Wegrzycki, Wlodarczyk'17, Künnemann, Paturi, Schneider'17].
This raised the question whether Knapsack can be solved in time $\tOh((n+\wmax)^2)$. This was open both for 0-1-Knapsack (where each item can be picked at most once) and Bounded Knapsack (where each item comes with a multiplicity).
The quest of resolving this question lead to algorithms that solve Bounded Knapsack in time $\tOh(n^3 \wmax^2)$ [Tamir'09], $\tOh(n^2 \wmax^2)$ and $\tOh(n \wmax^3)$ [Bateni, Hajiaghayi, Seddighin, Stein'18], $O(n^2 \wmax^2)$ and $\tOh(n \wmax^2)$ [Eisenbrand and Weismantel'18], $O(n + \wmax^3)$ [Polak, Rohwedder, Wegrzycki'21], and very recently $\tOh(n + \wmax^{12/5})$ [Chen, Lian, Mao, Zhang'23]. 

In this paper we resolve this question by designing an algorithm for Bounded Knapsack with running time $\tOh(n + \wmax^2)$, which is conditionally near-optimal. This resolves the question both for the classic 0-1-Knapsack problem and for the Bounded Knapsack problem.
\end{abstract}


\section{Introduction}

%

Knapsack is one of the most fundamental problems at the intersection of computer science, optimization, and operations research. It appeared as one of Karp's original 21 NP-hard problems~\cite{Karp72} and has been subject to an extensive amount of research, see, e.g., the book~\cite{KPP04book}. 

In the 0-1-Knapsack problem we are given a weight budget $W$ and $n$ items, and for each item~$i$ we are given its weight $w_i$ and its profit $p_i$. The goal is to select a set of items of total weight at most $W$ and maximum total profit. Formally, 0-1-Knapsack is the following problem:
$$ \max\{ p^T x : w^T x \le W,\, x \in \{0,1\}^n \}. $$

A generalization of 0-1-Knapsack is the Bounded Knapsack problem, in which each item $i$ also comes with a multiplicity $u_i$, and item $i$ can be selected up to $u_i$ times. This can be viewed as a compressed representation of a 0-1-Knapsack instance. In particular, any algorithm for Bounded Knapsack also solves 0-1-Knapsack. Formally, Bounded Knapsack is the following problem:
$$ \max\{ p^T x : w^T x \le W,\, 0 \le x \le u,\, x \in \mathbb{Z}^n \}. $$

There is a vast amount of literature on both problems.
When the input integers are small, of particular importance are pseudopolynomial-time algorithms, i.e., algorithms whose running time depends polynomially on $n$ and the input integers. 
A well-known example is Bellman's dynamic programming algorithm from 1957~\cite{bellman1957dynamic} that solves 0-1-Knapsack in time $O(n W)$ and Bounded Knapsack in time $\tOh(n W)$.\footnote{By $\tOh$-notation we hide logarithmic factors, i.e., $\tOh(T) = \bigcup_{c \ge 0} O(T \log^c T)$.} 
In the last few years, research on pseudopolynomial-time algorithms for Knapsack was driven by developments in fine-grained complexity theory (e.g., \cite{CyganMWW19,KunnemannPS17,AbboudBHS22}) and in proximity bounds for integer programming (e.g., \cite{EisenbrandW20,JansenR19}).
In particular, fine-grained complexity contributed conditional lower bounds for 0-1-Knapsack and Bounded Knapsack that rule out time $O((n+W)^{2-\delta})$~\cite{CyganMWW19,KunnemannPS17} and $W^{1-\delta} \cdot 2^{o(n)}$~\cite{AbboudBHS22} for any constant $\delta > 0$. That is, by now we have evidence that the running time $\tOh(n W)$ is near-optimal.

To cope with these hardness results, recent work studied the maximum weight of any item $\wmax$ as a parameter. This is especially interesting for Bounded Knapsack, where the parameter $W$ can be much larger than $\wmax$ and $n$ due to the multiplicities. 
For this reason, it is far from obvious that Bounded Knapsack can be solved in polynomial time in terms of $n$ and $\wmax$. The first such algorithm was designed by Tamir~\cite{Tamir09}, achieving time $\tOh(n^3 \wmax^3)$.
This raised the question: \emph{What is the optimal running time for Bounded Knapsack in terms of $n$ and $\wmax$?} More precisely, what is the smallest constant $c$ such that Bounded Knapsack can be solved in time $\tOh((n + \wmax)^c)$?
%
%
The same conditional lower bound that rules out time $O((n+W)^{2-\delta})$ also rules out time $O((n+\wmax)^{2-\delta})$ for any constant $\delta > 0$~\cite{CyganMWW19,KunnemannPS17}. Thus, we have $c \ge 2$, and the driving question becomes:


\begin{center}
\emph{Can Bounded Knapsack be solved in time $\tOh((n+\wmax)^2)$?}
\end{center}

This question motivated a line of research developing Bounded Knapsack algorithms with better and better dependence on $n$ and $\wmax$, see Table~\ref{tab:running-times}. In particular, after Tamir's first polynomial algorithm with time complexity $\tOh(n^3 \wmax^2)$~\cite{Tamir09}, a line of research developed and refined proximity bounds by exchange arguments leading to running times $\tOh(n^2 \wmax^2)$ and $\tOh(n \wmax^3)$~\cite{BateniHSS18}, $O(n^2 \wmax^2)$ and $\tOh(n \wmax^2)$~\cite{EisenbrandW20}, and $O(n + \wmax^3)$~\cite{PolakRW21}. Very recently, Chen et al.\ used exchange arguments based on additive combinatorics to obtain time $\tOh(n + \wmax^{2.4})$~\cite{ChenLMZ23}. 
However, despite extensive research our understanding of the optimal running time in terms of $n$ and $\wmax$ is still incomplete. In particular, the above driving question remained open.

\setlength{\tabcolsep}{10pt}
\bgroup
\def\arraystretch{1.3}
\begin{table}[ht]
    \caption{Pseudopolynomial-time algorithms for 0-1 Knapsack and Bounded Knapsack parameterized by the number of items $n$ and the maximum weight of any item $\wmax$.} \label{tab:running-times}
    \centering
    \begin{tabular}{|p{8cm}ll|}
        \hline
        Reference & Running Time & Problem Variant \\
        \hline
        Bellman~\cite{bellman1957dynamic} & $O(n^2 \wmax)$ & 0-1-Knapsack  \\
        Tamir~\cite{Tamir09} & $\tOh(n^3 \wmax^2)$ & Bounded Knapsack \\
        Bateni, Hajiaghayi, Seddighin and Stein~\cite{BateniHSS18} & $\tOh(n^2 \wmax^2)$ & Bounded Knapsack  \\
        Bateni, Hajiaghayi, Seddighin and Stein~\cite{BateniHSS18} & $\tOh(n \wmax^3)$ & Bounded Knapsack  \\
        \parbox{8cm}{Bateni, Hajiaghayi, Seddighin and Stein~\cite{BateniHSS18}, \\ Axiotis and  Tzamos~\cite{AxiotisT19} (and implicit in \cite{KellererP04})} & $O(n \wmax^2)$ & 0-1-Knapsack  \\
        Eisenbrand and Weismantel~\cite{EisenbrandW20} & 
$O(n^2 \wmax^2)$ & Bounded Knapsack \\
        Eisenbrand and Weismantel~\cite{EisenbrandW20} & 
$\tOh(n \wmax^2)$ & Bounded Knapsack \\
        Polak, Rohwedder and Wegrzycki~\cite{PolakRW21} & $O(n + \wmax^3)$ & Bounded Knapsack  \\
        Jin~\cite{Jin23} & $\tOh(n + \wmax^{2.5})$ & 0-1-Knapsack  \\
        Chen, Lian, Mao and Zhang~\cite{ChenLMZ23} & $\tOh(n + \wmax^{2.4})$ & Bounded Knapsack  \\
        \textbf{This work} & $\tOh(n + \wmax^{2})$ & Bounded Knapsack  \\
        Independent work by Jin~\cite{JinArxiv23} & $\tOh(n + \wmax^{2})$ & 0-1-Knapsack  \\
        \hline
    \end{tabular}
\end{table}
\egroup

For 0-1-Knapsack the state of the art is very similar to its generalization Bounded Knapsack. In addition to the algorithms listed above, the following algorithms were developed specifically for 0-1-Knapsack: Note that instances with $W \ge n \wmax$ are trivial for 0-1-Knapsack. Instances with $W < n \wmax$ can be solved in time $O(n W) \le O(n^2 \wmax)$ by Bellman's classic dynamic programming algorithm~\cite{bellman1957dynamic}, in time $O(n + \wmax W) \le O(n \wmax^2)$~\cite{BateniHSS18,AxiotisT19} (also implicit in~\cite{KellererP04}), and since recently in time $\tOh(n + \wmax^{2.5})$~\cite{Jin23}.
The same conditional lower bound as for Bounded Knapsack also works for 0-1-Knapsack, and thus for both problems our understanding had the same gap. 

In particular, the driving question posed above was open not only for Bounded Knapsack but also for 0-1-Knapsack. We remark that this driving question has been asked repeatedly (for Bounded Knapsack or 0-1-Knapsack), see, e.g.~\cite{PolakRW21,BringmannC22,Jin23,ChenLMZ23}.

%
%
%
%

\subsection{Our Results}

In this paper, we resolve the driving question by designing an algorithm for Bounded Knapsack that runs in time $\tOh(n + \wmax^2)$.


\begin{restatable}{theorem}{thmmain}
\label{thm:main}
  Bounded Knapsack can be solved by a deterministic algorithm in time $\tOh(n + \wmax^2)$.
\end{restatable}

This resolves the driving question of a long line of research~\cite{Tamir09,BateniHSS18,AxiotisT19,
EisenbrandW20,PolakRW21,Jin23,ChenLMZ23}. As mentioned above, our running time $\tOh(n + \wmax^2)$ is conditionally near-optimal; more precisely an algorithm with running time $O((n+\wmax)^{2-\delta})$ for any constant $\delta > 0$ would violate the $(\min,+)$-Convolution Conjecture~\cite{CyganMWW19,KunnemannPS17}. 
Our result adds to a small, but growing number of conditionally near-optimal pseudopolynomial-time algorithms for classic optimization problems, see the related work below.

\medskip
Alternatively, in Theorem~\ref{thm:main} we can replace $\wmax$ by the maximum profit of any item $\pmax$. 


\begin{restatable}{theorem}{thmmainprofit}
\label{thm:mainprofit}
  Bounded Knapsack can be solved by a deterministic algorithm in time $\tOh(n + \pmax^2)$.
\end{restatable}

%
%
%

\paragraph{Independent Work}
Jin~\cite{JinArxiv23} independently also achieved running time $\tOh(n + \wmax^2)$ as well as $\tOh(n + \pmax^2)$. Let us give a brief comparison: In Jin's favor, his $\tOh$ hides only 4 logfactors, while our $\tOh$ hides at least 7 logfactors. In our favor, (1) we handle Bounded Knapsack, while Jin's algorithm only works for 0-1-Knapsack, and (2) our algorithm seems to be significantly simpler.

\subsection{Technical Overview}


Here we describe the tools used in our algorithm and their history. 

\paragraph{Classic Proximity Bound}
Let $(w,p,W)$ be a 0-1-Knapsack instance sorted by non-increasing profit-to-weight ratios $p_1/w_1 \ge \ldots \ge p_n/w_n$. Let $g$ be the \emph{maximal prefix solution}, i.e., $g$ picks the maximal prefix of the items that fits into the weight budget $W$.  
Intuitively, an optimal solution~$x^*$ should not deviate too much from $g$, since any deviation means replacing an item with higher profit-to-weight ratio by an item with lower profit-to-weight ratio. Formally, a classic proximity bound shows that there exists an optimal solution $x^*$ that differs from the maximal prefix solution~$g$ in $O(\wmax)$ entries. That is, we have
\begin{align} \label{eq:standardproximity}
  \sum_{i \in [n]} |x^*_i - g_i| = O(\wmax) \quad\quad \text{and thus} \quad\quad \sum_{i \in [n]} w_i \cdot |x^*_i - g_i| = O(\wmax^2). \tag{CP}
\end{align}
This proximity bound was pioneered by Eisenbrand and Weismantel~\cite{EisenbrandW20}. Their proof also works for Integer Linear Programs with more than one constraint and is based on Steinitz' lemma; for Knapsack a simplified proof using a simple exchange argument can be found in~\cite{PolakRW21}.
This proximity bound has been used in almost all recent work on pseudopolynomial-time and approximation algorithms for Knapsack. 
We also make use of it, in order to reduce the Bounded Knapsack problem to the 0-1-Knapsack problem, as we explain next. 


\paragraph{Reduction from Bounded Knapsack to 0-1-Knapsack}
We show the following reduction between the two Knapsack problems: \emph{If 0-1-Knapsack on $O(\wmax^2)$ items can be solved in time $\tOh(\wmax^2)$, then Bounded Knapsack can be solved in time $\tOh(n + \wmax^2)$.} 
After establishing this reduction, to show our main result that Bounded Knapsack can be solved in time $\tOh(n + \wmax^2)$ it suffices to show that 0-1-Knapsack can be solved in time $\tOh(n + \wmax^2)$. 

We obtain this reduction along the lines of previous Bounded Knapsack algorithms, e.g.~\cite{EisenbrandW20,PolakRW21}. 
In the following description for simplicity think of the starting point being a 0-1-Knapsack instance (the same reduction can also be efficiently implemented when starting from Bounded Knapsack).
Given a Knapsack instance, we construct the maximal prefix solution $g$. 
For any weight $\hat w \in [\wmax]$, among the items of weight $\hat w$ that are picked by~$g$, we greedily pick all but the $\Theta(\wmax)$ least profitable items, and we remove the picked items from the instance. 
Similarly, among the items of weight~$\hat w$ that are not picked by $g$, we remove all but the $\Theta(\wmax)$ most profitable items from the instance. 
The classic proximity bound (\ref{eq:standardproximity}) implies that some optimal solution is consistent with these choices.
The remaining instance has at most $\Theta(\wmax)$ items of each weight in $[\wmax]$, so we reduced the number of items to $O(\wmax^2)$. So from now on we focus on 0-1-Knapsack with~few~items.



\paragraph{Additive Combinatorics}
A fundamental result in additive combinatorics shows that for any set $A \subseteq [n]$ of size $|A| \gg \sqrt{n}$ the set of all subset sums of $A$ contains an arithmetic progression of length~$n$. This result was pioneered by Freiman~\cite{Freiman93} and Sárközy~\cite{sarkozy94} and further improved by Szemerédi and Vu~\cite{szemeredi2006finite} and recently by Conlon, Fox, and Pham~\cite{conlon2021subset}. 
A long line of work used variations of this result to design algorithms for dense cases of the Subset Sum problem~\cite{chaimovich1999new,ChaimovichFG89,
freiman1988extremal,galil1991almost,GalilM91,
BringmannW21}. Notably, the majority of this line of research assumes that the input is a set, and only recently this work has been generalized to multi-sets~\cite{BringmannW21}. 
In this paper we distill the following tool from~\cite{BringmannW21}, which roughly states that two multi-sets of small integers $X,Y$ have two subsets of equal sum if both the size of $X$ and the sum of all elements of $Y$ are sufficiently large.
\begin{lemma}[Informal Version of Lemma~\ref{lem:addcomb}] \label{lem:addcombinformal}
  Let $X,Y$ be multi-sets consisting of integers in $[\wmax]$, where each number in $X$ has multiplicity at most $\mu$. If $|X| \gg (\mu \cdot \wmax)^{1/2}$ and $\sum_{y \in Y} y \gg \mu \cdot \wmax^2 / |X|$ then there exist non-empty subsets of $X$ and $Y$ that have the same sum.
\end{lemma}

\paragraph{Proximity Bounds based on Additive Combinatorics}
Results from additive combinatorics such as Lemma~\ref{lem:addcombinformal} can be used in exchange arguments to show proximity bounds. For Knapsack, this was first used by Deng, Jin and Mao to design approximation schemes~\cite{DengJM23}, and recently by Jin~\cite{Jin23} and Chen, Lian, Mao and Zhang~\cite{ChenLMZ23} to design pseudopolynomial-time algorithms. Our arguments follow the approach of Chen et al.~\cite{ChenLMZ23}. 

The basic exchange argument by Chen et al.\ is as follows. 
Assume for simplicity that all items have distinct profit-to-weight ratios, i.e., $p_1/w_n > \ldots > p_n/w_n$. Fix an optimal solution $x^*$, and consider a set of items $I_X$ that are picked by $x^*$ and a set of items $I_Y$ that are not picked by~$x^*$. Consider the multi-set of all weights of items in $I_X$, i.e., $X = \{ w_i : i \in I_X \}$, and similarly define $Y = \{ w_i : i \in I_Y \}$. If $X$ and $Y$ satisfy the constraints of Lemma~\ref{lem:addcombinformal}, then there exist subsets $I'_X \subset I_X$ and $I'_Y \subseteq I_Y$ of equal total weight. Then removing $I'_X$ and adding $I'_Y$ to $x^*$ maintains the total weight, and thus yields a new feasible solution $x'$. Now suppose that all items in $I_Y$ have a higher profit-to-weight ratio than all items in $I_X$; since we sorted by decreasing profit-to-weight ratio this condition is equivalent to $\max(I_Y) < \min(I_X)$. Then this exchange strictly increases the profit, contradicting the choice of $x^*$ as an optimal solution. We can use this contradiction to argue that one of the assumptions of Lemma~\ref{lem:addcombinformal} must not be satisfied, and by ensuring that $|X| = |I_X|$ is large, we can conclude that the sum of all weights of the items in $I_Y$ must be small. By picking $I_Y$ as a set of items picked by $g$ but not by $x^*$, we obtain that these items have a small total weight. By a symmetric argument on the items picked by $x^*$ but not by $g$, we obtain that $x^*$ and $g$ differ by a small total weight. This can yield a proximity bound similar to (\ref{eq:standardproximity}).

The question now is how to choose the set $I_Y$ and ensure that there always exists a corresponding set $I_X$ so that the basic exchange argument can be applied. 
To this end, Chen et al.\ partition the items into 6 parts, and show applicability on each part.
We devise a more intricate partitioning into a polylogarithmic number of parts, proving the following novel proximity bound. 
\begin{theorem}[Informal Version of Theorem~\ref{thm:partitioning}] \label{thm:partitioninginformal}
  Given a 0-1-Knapsack instance $(w,p,W)$ of size~$n$ with maximal prefix solution $g$, in time $\tOh(n)$ we can construct a partitioning $I_1 \cup \ldots \cup I_k = [n]$ into $k = O(\log^2 n)$ parts together with a proximity bound $\Delta_j$ for each part $I_j$ such that for every~$j$ we have (1) any optimal solution $x^*$ differs from $g$ by total weight at most $\Delta_j$ among the items $I_j$, and (2) the $\Delta_j$'s are small, more precisely the product of $\Delta_j$ and the number of distinct weights in~$I_j$ is $\tOh(\wmax^2)$.
\end{theorem}
This theorem is the main novelty of our algorithm. 
For comparison, note that the unique partitioning consisting of $k=1$ part satisfies $|\supp(w(I_1))| \cdot \Delta_1 \le \tOh(\wmax^{3})$, since $|\supp(w(I_1))| \le \wmax$ holds trivially and $\Delta_1 := \Theta(\wmax^2)$ works by the classic proximity bound (\ref{eq:standardproximity}). 
This bound for $k=1$ was recently improved to $\tOh(\wmax^{2.5})$ by Jin~\cite{Jin23}. 
Chen et al.\ were the first to study partitionings into $k>1$ parts, constructing a partitioning into 6 parts with $|\supp(w(I_j))| \cdot \Delta_j \le \tOh(\wmax^{2.4})$~\cite[Lemmas 13, 16 and 17]{ChenLMZ23}. We improve this bound to $\tOh(\wmax^2)$, while using more parts.



\paragraph{The Knapsack Algorithm}
Given Theorem~\ref{thm:partitioninginformal}, designing an algorithm that solves 0-1-Knapsack in time $\tOh(n + \wmax^2)$ follows from standard techniques. 
We split each part $I_j$ into the items $I_j^-$ picked by $g$ and the items $I_j^+$ not picked by $g$. 
By Theorem~\ref{thm:partitioninginformal}, any optimal solution $x^*$ for the whole instance $(w,p,W)$ uses weight at most $\Delta_j$ in the set $I^+_j$, and thus we can restrict our attention to weights $0 \le W' \le \Delta_j$ in set $I^+_j$. Similarly, $g - x^*$ uses weight at most $\Delta_j$ in the set~$I^-_j$, and thus we can restrict our attention to weights $0 \le W' \le \Delta_j$ in set $I^-_j$. For simplicity, in this overview we focus on the set $I^+_j$.
We further split each set $I^+_j$ according to item weights; this splits $I^+_j$ into sets $I^+_{j,1},\ldots,I^+_{j,\ell_j}$ where $\ell_j$ is the number of distinct weights among the items in $I^+_j$. 
All items in $I^+_{j,\ell}$ have the same weight, so a standard greedy algorithm computes the optimal profit for each weight $0 \le W' \le \Delta_j$ in time $O(\Delta_j)$. 
The total time to run the greedy algorithm on all sets $I^+_{j,1},\ldots,I^+_{j,\ell_j}$ is $O(\Delta_j)$ times the number of distinct weights of items in $I_j^+$, which by guarantee (2) of Theorem~\ref{thm:partitioninginformal} is $\tOh(\wmax^2)$. Repeating this for all $1 \le j \le k = O(\log^2 n)$ still runs in time $\tOh(\wmax^2)$.

It remains to combine the results of the greedy algorithm. Here we use that the result of the greedy algorithm is a concave sequence, and the $(\max,+)$-convolution of an arbitrary sequence and a concave sequence can be computed in linear time. It follows that combining the results of the greedy algorithm can be done in roughly the same running time as it took to compute them, up to logarithmic factors.
Hence, in total the running time is still $\tOh(n + \wmax^2)$.

We remark that (apart from Theorem~\ref{thm:partitioninginformal}) this algorithm is quite standard:
For $k = 1$ the same algorithm was described in~\cite[Lemma 2.2]{PolakRW21}, for larger $k$ essentially the same algorithm was described in~\cite[Lemma~6]{ChenLMZ23}, and its ingredients (i.e., solving equal weights by a greedy algorithm and combining concave sequences) have been used repeatedly, see, e.g.,~\cite{KellererP04,CH2218a,AxiotisT19,PolakRW21,
ChenLMZ23}.


\paragraph{Summary}
We described an algorithm that solves Bounded Knapsack in time $\tOh(n + \wmax^2)$ and thus resolves the driving question of a long line of research~\cite{Tamir09,BateniHSS18,AxiotisT19,
EisenbrandW20,PolakRW21,Jin23,ChenLMZ23}.
Our main novelty is a proximity bound that splits the set of items into a polylogarithmic number of parts, so that within each part the optimal solution cannot deviate much from the maximal prefix solution.
%
We highlight that additive combinatorics is a crucial tool in our algorithm, thus showing once more that the recent trend of using additive combinatorics for algorithm design is very fruitful.

\subsection{Further Related Work}


\paragraph{Further Parametrizations}
So far we discussed Knapsack algorithms that parametrize by the number of items $n$ and the weight parameters $W$ or $\wmax$. 
The analogous parameters for profits are $\textup{OPT}$ (the profit of an optimal solution) and $\pmax$ (the maximum profit of any item). It is essentially symmetric to parametrize by weight or by profit parameters, see, e.g., \cite[Section 4]{PolakRW21}. However, interesting new algorithms are possible when parametrizing by both weight and profit parameters, see the line of work~\cite{Pisinger99,BateniHSS18,BringmannC22,
BringmannC23}. To mention an exemplary result, 0-1-Knapsack can be solved in time $\tOh(n \wmax \pmax^{2/3})$~\cite{BringmannC23}. What we present in this paper is superior in case $\wmax \ll \pmax$, but worse in case $\wmax \approx \pmax \gg n$.


\paragraph{Unbounded Knapsack}
The Unbounded Knapsack problem is the same as Bounded Knapsack but with $u_i = \infty$ for each item $i$. That is, every item can be selected any number of times, which makes the problem much simpler. 
The same conditional lower bound as for Bounded Knapsack ruling out time $O((n+\wmax)^{2-\delta})$ for any $\delta > 0$ also works for this problem~\cite{CyganMWW19,KunnemannPS17}.
Running time $\tOh(n + \wmax^2)$ has been achieved for Unbounded Knapsack independently by Jansen and Rohwedder~\cite{JansenR19} and Axiotis and Tzamos~\cite{AxiotisT19}. Their algorithms are significantly simpler than our algorithm for Bounded Knapsack, in particular they do not require additive combinatorics. 
Chan and He further improved the running time to $\tOh(n \wmax)$~\cite{ChanH22}.


\paragraph{Subset Sum}
The Subset Sum problem is the special case of 0-1-Knapsack with $w_i = p_i$ for each item $i$. Here the weight bound $W$ is typically called the target $t$. Bellman's classic dynamic programming algorithm solves Subset Sum in time $O(nt)$~\cite{bellman1957dynamic}, which has been improved to $\tOh(n+t)$ \cite{Bring17,JW19}. A conditional lower bound rules out time $t^{1-\delta} \cdot 2^{o(n)}$ for any constant $\delta > 0$~\cite{AbboudBHS22}. 

Subset Sum with parameter $\wmax$ has also been studied recently. The problem can be solved in time $\tOh(n + \wmax^{1.5})$, which was first known when the input is a set (by combining \cite{BringmannW21} and \cite{Bring17}), and was recently shown when the input is allowed to be a multi-set~\cite{Jin23,ChenLMZ23}. There remains a gap to the best known conditional lower bound ruling out time $\wmax^{1-\delta} \cdot 2^{o(n)}$ for any constant $\delta > 0$~\cite{AbboudBHS22}.

\paragraph{Approximation Algorithms}
A recent line of work uses similar tools as for pseudopolynomial-time algorithms to study approximation schemes for Knapsack, Subset Sum, and Partition from a fine-grained perspective, see, e.g.,~\cite{BringmannC22,BringmannN21,CH2218a,DengJM23,
Jin19,MuchaW019,WuChen22}.

\subsection{Organization}

After a short preliminaries section, in Section~\ref{sec:addcomb} we present tools from additive combinatorics. 
The main novelty of this paper can be found in Section~\ref{sec:proximity}, where we develop our proximity bounds. Then in Section~\ref{sec:algorithms} we present our algorithmic ingredients. Note that Sections~\ref{sec:proximity} and~\ref{sec:algorithms} are written for 0-1-Knapsack. We justify this by presenting a reduction from Bounded Knapsack to 0-1-Knapsack in Section~\ref{sec:reduction}. Finally, in Section~\ref{sec:puttingtogether} we combine these ingredients to prove our main result.

\section{Preliminaries}

Throughout the paper we use the notation $\mathbb{N} = \{1,2,\ldots\}$ and $\mathbb{N}_0 = \{0,1,2,\ldots\}$. For $n \in \mathbb{N}$ we write $[n] = \{1,2,\ldots,n\}$. We use $\tOh$-notation to hide logarithmic factors, where $\tOh(T) := \bigcup_{c \ge 0} O(T \log^c T)$. Formally, the problems studied in this paper are defined as follows.

\begin{problem}[Bounded Knapsack]
  Given profits $p \in \mathbb{N}^n$, weights $w \in \mathbb{N}^n$, and multiplicities $u \in \mathbb{N}^n$, as well as a weight budget $W \in \mathbb{N}$, compute $\max\{ p^T x : w^T x \le W, x \in \mathbb{Z}^n, 0 \le x \le u\}$.  
\end{problem}

0-1-Knapsack is the same as Bounded Knapsack but all multiplicites are 1, i.e., $u_i = 1$ for all $i$.

\begin{problem}[0-1-Knapsack]
  Given profits $p \in \mathbb{N}^n$, weights $w \in \mathbb{N}^n$, and a weight budget $W \in \mathbb{N}$, compute $\max\{ p^T x : w^T x \le W, x \in \{0,1\}^n\}$.  
\end{problem}

Throughout the paper, we denote by $\wmax$ the maximum weight $\wmax = \max_i w_i$.

For a given 0-1-Knapsack instance $(w,p,W)$ of size $n$ we usually assume that the items are sorted according to non-increasing profit-to-weight ratio $p_1/w_1 \ge \ldots \ge p_n/w_n$.
Having fixed such an ordering, an important object in this paper is the \emph{maximal prefix solution} $g \in \{0,1\}^n$ that selects a maximal feasible prefix of all items. 
More precisely, $g$ is defined as follows. Let $t \in [n+1]$ be maximal such that $w_1+w_2+\ldots+w_{t-1} \le W$. Then let $g_1 = \ldots = g_{t-1} = 1$ and $g_t = \ldots = g_n = 0$.\footnote{We think of $g$ as being a greedy solution, therefore we use the letter $g$ to denote it. However, note that a reasonable greedy algorithm would try to pick additional items and thus not necessarily pick a prefix of items, so greedy solutions are generally not the same as the maximal prefix solution.}
Note that $g$ is a feasible solution, since it has weight $w^T g = w_1+w_2+\ldots+w_{t-1} \le W$. It is also maximal, in the sense that adding item~$t$ would make it infeasible. In particular, the maximal prefix solution has weight $w^T g \in (W - \wmax, W]$ (unless $W \ge w_1+\ldots+w_n$, in which case we have $t = n+1$). 
Clearly, the maximal prefix solution can be computed in time $\tOh(n)$. 

For a given Bounded Knapsack instance $(w,p,u,W)$ of size $n$ we also usually assume non-increasing profit-to-weight ratios $p_1/w_1 \ge \ldots \ge p_n/w_n$.
The maximal prefix solution $g \in \mathbb{N}_0^n$ can then be defined analogously to the 0-1 case, by letting $t \in [n+1]$ be maximal such that $\sum_{i=1}^{t-1} u_i \cdot w_i \le W$ and setting $g_1 := u_1,\ldots, g_{t-1} := u_{t-1}$ and $g_{t+1} = \ldots = g_n := 0$ and $g_t := \lfloor (W - \sum_{i=1}^{t-1} u_i \cdot w_i) / w_t \rfloor$ (if $t \le n$). Again, $g$ is feasible and maximal and can be easily computed in time $\tOh(n)$.

\section{Tools from Additive Combinatorics}
\label{sec:addcomb}

A long line of work used additive combinatorics to design algorithms for dense cases of the Subset Sum problem~\cite{chaimovich1999new,ChaimovichFG89,
freiman1988extremal,galil1991almost,GalilM91,
BringmannW21}. The majority of this line of research works under the assumption that the input is a set; only recently this work has been generalized to multi-sets~\cite{BringmannW21}. In this section, we distill a tool from~\cite{BringmannW21} that we will later use in an exchange argument.
In what follows, after introducing the necessary notation, we phrase this tool in Lemma~\ref{lem:addcomb}, and then we show how to obtain Lemma~\ref{lem:addcomb} from~\cite{BringmannW21}.

\medskip
Let $X$ be a \emph{finite multi-set of positive integers}, or \emph{multi-set} for short. For an integer $x$ we denote by $\mu(x;X)$ the multiplicity of $x$ in $X$, indicating that the number $x$ appears $\mu(x;X)$ times in $X$. The usual set notation generalizes naturally to multi-sets, e.g., we write $X \subseteq Y$ if $\mu(x;X) \le \mu(x;Y)$ holds for all $x$. 
Furthermore, we define the following concepts:
\begin{itemize}
\item \emph{Support $\supp(X)$:} The set of all distinct integers contained in $X$, i.e., $\supp(X) = \{x \in \mathbb{N} \mid \mu(x;X) \ge 1\}$.
\item \emph{Size $|X|$:} The number of elements of $X$, counted with multiplicity, i.e., $|X| = \sum_{x \in \mathbb{N}} \mu(x;X)$.
\item \emph{Maximum $\max(X)$:} The maximum number in $X$, i.e., $\max(X) = \max\{x \in \mathbb{N} \mid \mu(x;X) \ge 1\}$.
\item \emph{Multiplicity $\mu(X)$:} The maximum multiplicity of $X$, i.e., $\mu(X) = \max_{x \in \mathbb{N}} \mu(x;X)$.
\item \emph{Sum $\Sigma(X)$:} The sum of all elements of $X$, i.e., $\Sigma(X) = \sum_{x \in \mathbb{N}} x \cdot \mu(x;X)$. 
\end{itemize}

\begin{lemma}[Additive Combinatorics Tool] \label{lem:addcomb}
  Let $X,Y$ be multi-sets consisting of integers in $[\wmax]$.~If 
  $$ |X| \ge 1500 \cdot (\log^3(2|X|) \mu(X) \wmax)^{1/2} \quad \text{and} \quad \Sigma(Y) \ge 340000 \log(2 |X|) \mu(X) \wmax^2 / |X|, $$ 
  then there exist non-empty $X' \subseteq X$ and $Y' \subseteq Y$ with equal sum $\Sigma(X') = \Sigma(Y')$.
\end{lemma}

We remark that \cite[Lemmas 11 and 12]{ChenLMZ23} are similar to Lemma~\ref{lem:addcomb} and are also derived from~\cite{BringmannW21}, but the specific formulations differ. 
The result is also similar in spirit to \cite[Lemma 3.4]{Jin23}.

The proof builds on the following definitions and theorems by Bringmann and Wellnitz~\cite{BringmannW21}. 

\begin{definition}[Definitions 3.1 and 3.2 in \cite{BringmannW21}]
  Let $X$ be a multi-set.
  Write $X(d) := X \cap d \mathbb{N}$ to denote the multi-set of all numbers in $X$
that are divisible by $d$. Further, write $\overline{X(d)} := X \setminus X(d)$ to denote the multi-set of all numbers in $X$
not divisible by $d$. An integer $d > 1$ is called an \emph{$\alpha$-almost divisor} of $X$ if $|\overline{X(d)}| \le \alpha \cdot \mu(X) \cdot \Sigma(X) / |X|^2$.

  If all integers in $X$ are divisible by $d$, we denote by $X/d$ the multi-set where we divide each element of $X$ by $d$, i.e., $\mu(x;X/d) = \mu(d \cdot x; X)$ for all $x$.

  The multi-set $X$ is called \emph{$\delta$-dense} if $|X|^2 \ge \delta \cdot \mu(X) \cdot \max(X)$.
\end{definition}

\begin{theorem}[Theorem 4.1 in \cite{BringmannW21}] \label{thm:BWone}
  Let $X$ be a multi-set and let $\delta \ge 1$ and $0 < \alpha \le \delta/16$.
  If $X$ is $\delta$-dense, then there exists an integer $d \ge 1$ such that $X' := X(d)/d$ is $\delta$-dense and has no $\alpha$-almost divisor. 
  Moreover, we have the following additional properties: 
  \begin{enumerate}
    \item $d \le 4 \mu(X) \Sigma(X) / |X|^2$,
    \item $|X'| \ge 0.75\, |X|$, 
    \item $\Sigma(X') \ge 0.75\, \Sigma(X)/d$.
  \end{enumerate}
\end{theorem}

\begin{theorem}[Theorem 4.2 in \cite{BringmannW21}] \label{thm:BWtwo}
  Let $X'$ be a multi-set and define the following parameters:
  \begin{align*}
    C_\delta(X') &:= 1699200 \cdot \log (2|X'|) \log^2(2\mu(X')), \\
    C_\alpha(X') &:= 42480 \cdot \log(2\mu(X')), \\
    C_\lambda(X') &:= 169920\cdot \log(2\mu(X')), \\
    \lambda(X') &:= C_\lambda(X') \cdot \mu(X') \max(X') \Sigma(X') / |X'|^2.
  \end{align*}
  If $X'$ is $C_\delta(X')$-dense and has no $C_\alpha(X')$-almost divisor,
  then for every integer $s$ in the interval $[\,\lambda(X'),\,\Sigma(X') - \lambda(X')\,]$ there exists $X'' \subseteq X'$ with $\Sigma(X'') = s$.
\end{theorem}

We are now ready to prove our main additive combinatorics tool, Lemma~\ref{lem:addcomb}.

\begin{proof}[Proof of Lemma~\ref{lem:addcomb}]
  Note that the assumption $|X| \ge 1500 \cdot (\log^3(2|X|) \mu(X) \wmax)^{1/2}$ implies that $|X| \ge \sqrt{C_\delta(X) \mu(X) \wmax}$, since $\sqrt{1699200} < 1500$ and $\mu(X) \le |X|$. Hence, $X$ is $C_\delta(X)$-dense. By Theorem~\ref{thm:BWone}, there exists an integer $d \ge 1$ such that $X' := X(d)/d$ is $C_\delta(X)$-dense and has no $C_\alpha(X)$-almost divisor. Since $|X'| \le |X|$ and $\mu(X') \le \mu(X)$, it follows that $X'$ is also $C_\delta(X')$-dense and has no $C_\alpha(X')$-almost divisor.
  It then follows from Theorem~\ref{thm:BWtwo} that every integer in the range $[\,\lambda(X'),\,\Sigma(X') - \lambda(X')\,]$ is a subset sum of $X'$. Hence, every multiple of $d$ in the range $[\,d \cdot \lambda(X'),\,d(\Sigma(X') - \lambda(X'))\,]$ is a subset sum of $X$. In what follows we need bounds on $d \cdot \lambda(X')$ and $d(\Sigma(X') - \lambda(X'))$, which we prove in the next claim.

  \begin{claim} \label{cla:intervalbounds}
    We have $d \cdot \lambda(X') \le \Sigma(Y) - d \wmax$ and $d(\Sigma(X') - 2 \lambda(X')) \ge d \wmax$.
  \end{claim}
  \begin{proof}
    From Theorem~\ref{thm:BWone} we obtain $C_\lambda(X') \le C_\lambda(X)$ and $\mu(X') \le \mu(X)$ and $\max(X') \le \max(X)/d \le \wmax/d$ and $\Sigma(X') \le \Sigma(X)/d$, and furthermore $|X'| \ge 0.75 |X|$. We thus have
    $$ \lambda(X') = C_\lambda(X') \cdot \mu(X') \max(X') \Sigma(X') / |X'|^2 \le C_\lambda(X) \cdot \mu(X) \wmax \Sigma(X) / (0.75 |X| \cdot d)^2. $$
    Since $0.75^2 > 1/2$ and $d \ge 1$ and $\Sigma(X) \le \wmax |X|$ we obtain
    \begin{align} \label{eq:sfhohg}
      d \cdot \lambda(X') \le 2 C_\lambda(X) \cdot \mu(X) \wmax \Sigma(X) / |X|^2 \le 2 C_\lambda(X) \cdot \mu(X) \wmax^2 / |X|.
    \end{align}
    By Theorem~\ref{thm:BWone} we have $d \le 4 \mu(X) \Sigma(X) / |X|^2 \le 4 \mu(X) \wmax / |X|$, and thus $d \wmax \le 4 \mu(X) \wmax^2 / |X|$. Hence, we have
    \begin{align*}
    d \cdot \lambda(X') + d \wmax 
    &\le (2 C_\lambda(X) + 4) \cdot \mu(X) \wmax^2 / |X|   \\
    &\le 340000 \log(2|X|) \mu(X) \wmax^2 / |X| \le \Sigma(Y), 
    \end{align*}
    where the last inequality holds by assumption of the theorem statement. This proves the first claim $d \cdot \lambda(X') \le \Sigma(Y) - d \wmax$.
    
    By Theorem~\ref{thm:BWone} we have $\Sigma(X') \ge 0.75 \Sigma(X) / d$. Using the first inequality of (\ref{eq:sfhohg}) now yields
    $$ d(\Sigma(X') - 2 \lambda(X')) \ge 0.75 \Sigma(X) - 4 C_\lambda(X) \cdot \mu(X) \wmax \Sigma(X) / |X|^2. $$
    Using $|X|^2 \ge 1500^2 \cdot \log^3(2|X|) \mu(X) \wmax$ (from the theorem assumption) now yields
    $$ d(\Sigma(X') - 2 \lambda(X')) \ge (0.75 \cdot 1500^2 \cdot \log^3(2|X|) - 4 C_\lambda(X)) \mu(X) \wmax \Sigma(X) / |X|^2. $$
    Since $C_\lambda(X) \le 170000 \log(2|X|)$ we obtain
    $$ d(\Sigma(X') - 2 \lambda(X')) \ge 1000000 \log^3(2|X|) \mu(X) \wmax \Sigma(X) / |X|^2 \ge 4 \mu(X) \wmax \Sigma(X) / |X|^2. $$
    Finally, Theorem~\ref{thm:BWone} yields $d \le 4 \mu(X) \Sigma(X) / |X|^2$ and thus $d(\Sigma(X') - 2 \lambda(X')) \ge d \wmax$.
  \end{proof}
  
  Now consider the multi-set $Y$. 
  Suppose that $|Y| \ge d$ and pick any elements $y_1,\ldots,y_d \in Y$. Consider the prefix sums $s_i := \sum_{j=1}^i y_j$ for $i=0,1,\ldots,d$. By the pigeonhole principle, two prefix sums are equal modulo $d$; say we have $s_i = s_j$ for $i < j$. Then $s_j - s_i = y_{i+1} + y_{i+2} + \ldots + y_j$ is divisible by $d$. Moreover, we have $y_{i+1} + y_{i+2} + \ldots + y_j \le (j-i) \wmax \le d \wmax$. We have thus shown that there exists $Y' \subseteq Y$ whose sum $\Sigma(Y')$ is divisible by $d$ and bounded from above by $d \wmax$, assuming that $|Y| \ge d$.
  
  Repeatedly find such a multi-set $Y' \subseteq Y$, remove its elements from $Y$, and continue on the remaining multi-set $Y$. This works until $Y$ has less than $d$ elements remaining. This process yields multi-sets $Y_1,\ldots,Y_k \subseteq Y$ whose sums $\Sigma(Y_1),\ldots,\Sigma(Y_k)$ are divisible by $d$ and bounded from above by $d \wmax$. Moreover, since less than $d$ elements remain, their total sum is $\Sigma(Y_1) + \ldots + \Sigma(Y_k) \ge \Sigma(Y) - d \wmax \ge d \cdot \lambda(X')$ by Claim~\ref{cla:intervalbounds}. 
  Since by Claim~\ref{cla:intervalbounds} the interval $[\,d \cdot \lambda(X'),\,d(\Sigma(X') - \lambda(X'))\,]$ has length at least $d \wmax$, 
  it follows that one of the prefix sums $\Sigma(Y_1) + \ldots + \Sigma(Y_\ell)$ for some $\ell$ lies in the interval $[\,d \lambda(X'),\,d(\Sigma(X') - \lambda(X'))\,]$ and is divisible by $d$. Since $\Sigma(Y_1) + \ldots + \Sigma(Y_\ell)$ is a subset sum of $Y$, and since every integer in $[\,d \cdot \lambda(X'),\,d(\Sigma(X') - \lambda(X'))\,]$ that is divisible by $d$ is a subset sum of $X$, it follows that there exists a number that is a subset sum of both $X$ and $Y$.
\end{proof}

\section{Proximity Bound} \label{sec:proximity}

The main result of this section is the following theorem, which provides a partitioning of the items of a 0-1-Knapsack instance into few parts, along with a certain proximity bound for each part.
In more detail, the theorem states that we can partition the items into a polylogarithmic number of sets $I_1,\ldots,I_k$ with corresponding proximity bounds $\Delta_1,\ldots,\Delta_k$ such that (1) inside each part $I_j$ the optimal solution deviates from the maximal prefix solution by a total weight of at most $\Delta_j$, and (2) the $\Delta_j$'s are not too large, in particular the product of $\Delta_j$ and the number of distinct weights in $I_j$ is bounded by $\tOh(\wmax^2)$.

In what follows for a 0-1-Knapsack instance $(w,p,W)$ and an index set $I \subseteq [n]$ we denote by $w(I)$ the multi-set $\{w_i : i \in I\}$, and by $\supp(w(I))$ the support of this multi-set, i.e., the set of distinct weights appearing among the items in $I$.

\begin{theorem}[Item Partitioning with Proximity Bounds] \label{thm:partitioning}
  For a 0-1-Knapsack instance $(w,p,W)$ of size $n$ with distinct profit-to-weight ratios $p_1/w_1 > \ldots > p_n/w_n$, let $g$ be the maximal prefix solution, and let $x^*$ be an arbitrary optimal solution.
  Given $w,p,W$, in time $\tOh(n)$ we can compute a partitioning $I_1 \cup \ldots \cup I_k = [n]$ and proximity bounds $\Delta_1,\ldots,\Delta_k$ with $k = O(\log^2 n)$ such that for each $1 \le j \le k$ we have
  $$ \sum_{i \in I_j} w_i |g_i - x^*_i| \le \Delta_j \quad \text{and} \quad |\supp(w(I_j))| \cdot \Delta_j \le 300000000 \log^3(2n) \cdot \wmax^2. $$
\end{theorem}

We remark that this theorem is the main novelty of our algorithm. The classic proximity bound (\ref{eq:standardproximity}) can be used to show that the unique partitioning consisting of $k=1$ part satisfies $|\supp(w(I_j))| \cdot \Delta_j \le \tOh(\wmax^{3})$. This bound was recently improved to $\tOh(\wmax^{2.5})$~\cite{Jin23}. 
Chen et al.\ were the first to study partitionings into $k>1$ parts. They constructed a partitioning into at most 6 parts with $|\supp(w(I_j))| \cdot \Delta_j \le \tOh(\wmax^{2.4})$~\cite[Lemmas 13, 16 and 17]{ChenLMZ23}. We improve this construction to $\tOh(\wmax^2)$.

\medskip
The rest of this section is devoted to the proof of Theorem~\ref{thm:partitioning}. 
We start by presenting algorithm \textsc{SingleStep}$(U)$, which is given a set $U \subseteq [n]$ and constructs a set $I \subseteq U$ and a distance bound~$\Delta$ with the same guarantees as in Theorem~\ref{thm:partitioning} (see Lemma~\ref{lem:constronestepguarantee} below). Later we will repeatedly apply algorithm \textsc{SingleStep} to obtain the partitioning promised by Theorem~\ref{thm:partitioning}. 
Algorithm \textsc{SingleStep}$(U)$ works as follows:
%
%
%
%

\begin{algorithm}[H]
    \caption{\textsc{SingleStep}$(U)$, with input $U \subseteq [n]$}\label{alg:SingleStep}
    \begin{algorithmic}[1]
        \State We consider the class of weights of largest multiplicity, defined as follows.
For every $\hat m \in [2n]$ that is a power of 2, we define 
$$ \cW_{\hat m} := \big\{\hat w \in [\wmax] : \hat m / 2 < |\{i \in U : w_i = \hat w\}| \le \hat m \big\}. $$
Let $m = m_U$ be the largest power of 2 such that $\cW_m$ is non-empty, and let $\cW := \cW_m$. \label{step:singlestep_one}

\State We let $J = J_U = \{i \in U:  w_i \in \cW \}$ be the items with weights in $\cW$. We split $J$ into the items picked by $g$ and the items unpicked by $g$: \label{step:singlestep_two}
$$ J^- := \{i \in J : g_i = 1 \}, \quad J^+ := \{i \in J : g_i = 0 \}. $$


\State We define $I^- \subseteq J^-$ to consist of the smallest $\lceil |J^-|/2 \rceil$ indices in $J^-$ (i.e., the items with largest profit-to-weight ratio), and $I^+ \subseteq J^+$ to consist of the largest $\lceil |J^+|/2 \rceil$ indices in $J^+$ (i.e., the items with smallest profit-to-weight ratio). Finally, we define $I = I_U$ to be the larger of the sets $I^-$ and $I^+$.

\State We set the distance bound to $\Delta = \Delta_U := 36000000 \log^3(2n) \cdot m \, \wmax^2 / |I|$.
    \end{algorithmic}
\end{algorithm}

In what follows we analyze this algorithm. 
We will use the notation introduced in the above algorithm throughout the rest of this section.
We start with a simple observation on the size of $I$.

\begin{obs} \label{obs:sizeI}
  We have $|I| \ge |J|/4$.
\end{obs}
\begin{proof}
  Note that
  $$ |I| \ge (|I^-| + |I^+|)/2 \ge (|J^-| + |J^+|)/4 = |J|/4, $$
  where we first used that $I$ is the larger of $I^-,I^+$, then that $I^-$ is at least half of $J^-$ and $I^+$ is at least half of $J^+$, and finally that $J^- \cup J^+ = J$ is a partitioning.
\end{proof}

We now observe the following upper bound on $\Delta$ times the number of distinct weights among~$I$.

\begin{lemma} \label{lem:easyproximity}
  We have $|\supp(w(I))| \cdot \Delta \le 300000000 \log^3(2n) \cdot \wmax^2$.
\end{lemma}
\begin{proof}
  Note that $|J| \ge \frac m2 \cdot |\cW|$ holds since all weights from $\cW$ appear at least $m/2$ times in $U$, and thus also in $J$. 
  Together with Observation~\ref{obs:sizeI} we obtain $|I| \ge |J|/4 \ge \frac m8 |\cW|$. 
  Since $\supp(w(I)) \subseteq \cW$, this yields $|\supp(w(I))| \le 8 |I| / m$. By the definition of $\Delta$ we now have
  $$ |\supp(w(I))| \cdot \Delta \le \frac {8 |I|}{m} \cdot 36000000 \log^3(2n) \cdot \frac {m \, \wmax^2}{|I|} \le 300000000 \log^3(2n) \cdot \wmax^2. \qedhere $$
\end{proof}

The following auxiliary lemma prepares the proof of the main proximity guarantee.


\begin{lemma} \label{lem:noequalsubsets}
  There are no non-empty sets $A \subseteq \{ i \in [n] : x^*_i = 0 \}$ and $B \subseteq \{ i \in [n] : x^*_i = 1 \}$ with 
  $\max(A) < \min(B)$ that have the same sum of weights $\Sigma(w(A)) = \Sigma(w(B))$.
\end{lemma}
\begin{proof}
  Towards a proof by contradiction, assume that sets $A,B$ as in the lemma statement exist.
  Let $x'$ be the solution obtained from~$x^*$ by removing the items $B$ and adding the items $A$, i.e., $\{ i \in [n] : x'_i = 1\} = (\{ i \in [n] : x^*_i = 1\} \setminus B) \cup A$.
  Since $\Sigma(w(A)) = \Sigma(w(B))$, we maintain the total weight $w^T x^* = w^T x'$, so the solution $x'$ is feasible.
  Since $\max(A) < \min(B)$, the profit-to-weight ratio $p_i/w_i$ of any item $i \in A$ is strictly larger than the profit-to-weight ratio $p_j/w_j$ of any item in $j \in B$, which yields $p^T x' > p^T x^*$. This contradicts $x^*$ being an optimal solution.
\end{proof}

We can now prove the main guarantee of algorithm \textsc{SingleStep}$(U)$: On the items in $I$ the optimal solution $x^*$ cannot deviate too much from the maximal prefix solution $g$. The proof of this proximity result is based on the additive combinatorics tool from the previous section.

\begin{lemma} \label{lem:constronestepguarantee}
  We have $\sum_{i \in I} w_i |g_i - x^*_i| \le \Delta$.
\end{lemma}

\begin{proof}
  We split the proof into several cases.
  
  \medskip
  \emph{Case 1: $|I| \le 6000 (\log^3(2n) \cdot m \, \wmax)^{1/2}$.}
  We square this inequality and rearrange it to $|I| \le 36000000 \log^3(2n) \cdot m \, \wmax / |I|$. Then we observe the trivial inequality $\sum_{i \in I} w_i |g_i - x^*_i| \le \wmax |I|$. Combining these two inequalities yields $\sum_{i \in I} w_i |g_i - x^*_i| \le 36000000 \log^3(2n) \cdot m \, \wmax^2 / |I| = \Delta$.
  
  \medskip
  \emph{Case 2: $|I| > 6000 (\log^3(2n) \cdot m \, \wmax)^{1/2}$.} We split this case into further subcases as follows.
  
  \medskip
  \emph{Case 2.1: $I = I^+$.} 
  
  \medskip
  \emph{Case 2.1.1: $\sum_{i \in J^+ \setminus I^+} |g_i - x^*_i| \le |J^+ \setminus I^+|/2$.} 
  Let $I_X = \{ i \in J^+ \setminus I^+ : x^*_i = 0 \}$ and consider its multi-set of weights $X = w(I_X)$. 
  We claim that $X$ satisfies the assumption of Lemma~\ref{lem:addcomb}, i.e., $|X| \ge 1500 (\log^3(2 |X|) \cdot \mu(X) \, \wmax)^{1/2}$. We prove this claim in the remainder of this paragraph. Note that $|X| = |I_X| \ge |J^+ \setminus I^+|/2$, since all items $i \in J^+ \setminus I^+$ have $g_i = 0$ and thus by the assumption of Case 2.1.1 at least half of these items have $x^*_i = 0$. Since $|I^+| = \lceil |J^+|/2 \rceil$ we have $|J^+ \setminus I^+| \ge |I^+|-1 = |I|-1 \ge |I|/2$, where we used $|I| > 1$ (by Case 2). 
  This yields 
  $$ |X| \ge |I|/4 \ge 1500 (\log^3(2n) \cdot m \, \wmax)^{1/2} \ge 1500 (\log^3(2 |X|) \cdot \mu(X) \, \wmax)^{1/2}, $$
  where we used the assumption of Case 2, $|X| \le n$, and $\mu(X) \le m$. Thus, $X$ satisfies the assumption of Lemma~\ref{lem:addcomb}.
  
  Let $I_Y := \{ i \in I^+ : x^*_i = 1 \}$ and consider its multi-set of weights $Y = w(I_Y)$. 
  Suppose that $\Sigma(Y)$ satisfies the assumption of Lemma~\ref{lem:addcomb}. Then we can apply Lemma~\ref{lem:addcomb} to $X$ and $Y$, which yields non-empty $X' \subseteq X$ and $Y' \subseteq Y$ with equal sum $\Sigma(X') = \Sigma(Y')$. 
  The corresponding subsets $I_{X'} \subseteq I_X$ and $I_{Y'} \subseteq I_Y$ have equal total weight $\Sigma(w(I_{X'})) = \Sigma(w(I_{Y'}))$. Since $\max(I_X) < \min(I_Y)$, this contradicts Lemma~\ref{lem:noequalsubsets}. 
  
  Therefore, $\Sigma(Y)$ cannot satisfy the assumption of Lemma~\ref{lem:addcomb}, i.e., we have 
  $$ \Sigma(Y) \le 340000 \log(2|X|) \mu(X) \wmax^2 / |X|. $$
  Now we again use $|I|/4 \le |X| \le n$ and $\mu(X) \le m$ to obtain
  $$ \Sigma(Y) \le 1360000 \log(2n) \cdot m \, \wmax^2 / |I| \le \Delta. $$
  Note that since $I = I^+$, all items $i \in I$ have $g_i = 0$ and thus $\Sigma(Y) = \Sigma(w(I_Y)) = \sum_{i \in I} w_i |g_i - x^*_i|$.
  We thus obtain $\sum_{i \in I} w_i |g_i - x^*_i| \le \Delta$, as desired. 
  
  \medskip
  \emph{Case 2.1.2: $\sum_{i \in J^+ \setminus I^+} |g_i - x^*_i| > |J^+ \setminus I^+|/2$.} 
  Let $I_X = \{ i \in J^+ \setminus I^+ : x^*_i = 1 \}$ and consider its multi-set of weights $X = w(I_X)$. 
  We claim that $X$ satisfies the assumption of Lemma~\ref{lem:addcomb}, i.e., $|X| \ge 1500 (\log^3(2 |X|) \cdot \mu(X) \, \wmax)^{1/2}$. We prove this claim in the remainder of this paragraph. Note that $|X| = |I_X| \ge |J^+ \setminus I^+|/2$, since all items $i \in J^+ \setminus I^+$ have $g_i = 0$ and thus by the assumption of Case 2.1.2 at least half of these items have $x^*_i = 1$. 
  The rest of the argument is exactly as in the previous case, and we arrive at $|X| \ge 1500 (\log^3(2 |X|) \cdot \mu(X) \, \wmax)^{1/2}$, so $X$ satisfies the assumption of Lemma~\ref{lem:addcomb}.
  
  Let $I_Y := \{ i \in [n] : g_i = 1, x^*_i = 0 \}$ and consider its multi-set of weights $Y = w(I_Y)$. 
  Suppose that $\Sigma(Y)$ satisfies the assumption of Lemma~\ref{lem:addcomb}. Then we can apply Lemma~\ref{lem:addcomb} to $X$ and $Y$, which yields $X' \subseteq X$ and $Y' \subseteq Y$ with equal sum $\Sigma(X') = \Sigma(Y')$. 
  The corresponding subsets $I_{X'} \subseteq I_X$ and $I_{Y'} \subseteq I_Y$ have equal total weight $\Sigma(w(I_{X'})) = \Sigma(w(I_{Y'}))$. Since $\max(I_Y) < \min(I_X)$, this contradicts Lemma~\ref{lem:noequalsubsets}. 
  
  Therefore, $\Sigma(Y)$ cannot satisfy the assumption of Lemma~\ref{lem:addcomb}. By the same calculations as in the previous case we arrive at
  $$ \Sigma(Y) \le 1360000 \log(2n) \cdot m \, \wmax^2 / |I| \le \Delta. $$
  It remains to relate $\Sigma(Y)$ to $\sum_{i \in I} w_i |g_i - x^*_i|$. Note that $x^*$ is maximal in the sense that no item can be added to it, and $g$ is maximal in the sense that it selects a maximal prefix of the items. It follows that either $w^T x^* = w^T g = \sum_{i \in [n]} w_i$ or
  $w^T x^*, w^T g \in (W-\wmax, W]$. Both yield
  $|w^T x^* - w^T g| < \wmax$. Moreover, we have
  $$ w^T x^* - w^T g = \sum_{i \in [n], g_i = 0} w_i |g_i - x^*_i| - \sum_{i \in [n], g_i = 1} w_i |g_i - x^*_i|,  $$
  and thus
  $$ \sum_{i \in I} w_i |g_i - x^*_i| \le \sum_{i \in [n], g_i = 0} w_i |g_i - x^*_i| < \sum_{i \in [n], g_i = 1} w_i |g_i - x^*_i| + \wmax. $$
  Finally, we note that
  $$ \Sigma(Y) = \Sigma(w(I_Y)) = \sum_{i \in [n], g_i = 1} w_i |g_i - x^*_i|, $$
  which implies the desired
  $\sum_{i \in I} w_i |g_i - x^*_i| \le \Sigma(Y) + \wmax \le \Delta$ (where we used $m \, \wmax \ge |I|$).

  \medskip
  \emph{Case 2.2: $I = I^-$.} This is symmetric to Case 2.1. To ensure that there are no subtle different details, we provide the full proof in Appendix~\ref{sec:appproofproximity}. 
\end{proof}

\begin{lemma} \label{lem:constrtime}
  Algorithm \textsc{SingleStep}$(U)$ runs in time $\tOh(n)$.
\end{lemma}
\begin{proof}
  This is straightforward. With one pass over all items in $U$ we can determine the set of distinct weights of all items. With another pass we can count for each distinct weight the number of items having that weight. This yields the set $\cW$. With another pass over all items in $U$ we can then determine $J^-, J^+$, and from these sets we can easily read off $I^-, I^+$, which yields $I$ and $\Delta$.
\end{proof}

We are now ready to prove our main proximity result Theorem~\ref{thm:partitioning}.

\begin{proof}[Proof of Theorem~\ref{thm:partitioning}]
  Starting with $U_1 = [n]$, we repeatedly run $\textsc{SingleStep}(U_j)$ to obtain $(I_j,\Delta_j)$ and remove $I_j$ from $U_j$ to obtain $U_{j+1}$, until $U_{j+1} = \emptyset$. For details, see the following pseudocode.

  
  \begin{algorithm}[H]
    \caption{Construction of the partitioning  guaranteed by Theorem~\ref{thm:partitioning}}\label{alg:partitioning}
    \begin{algorithmic}[1]
    \State $U_1 := [n]$
    \For{$j=1,2,\ldots$}:
      \State $(I_j,\Delta_j) := \textsc{SingleStep}(U_j)$
      \State $U_{j+1} := U_j \setminus I_j$
      \State If $U_{j+1} = \emptyset$ then return $(I_1,\Delta_1),\ldots,(I_j,\Delta_j)$
    \EndFor
    \end{algorithmic}
  \end{algorithm}
  
  To see that this procedure terminates, observe that for non-empty $U$ algorithm $\textsc{SingleStep}(U)$ computes a non-empty set $I$. Thus, the procedure terminates after at most $n$ iterations. We denote the number of iterations until termination by $k$.
  
  It is easy to see that $I_1, \ldots, I_k$ form a partitioning of $[n]$, since (1) in each iteration $j$ the items $I_j$ are selected from the remaining items $U_j = [n] \setminus (I_1 \cup \ldots \cup I_{j-1})$, and (2) in the end there are no more remaining items, so we have $\emptyset = U_{k+1} = [n] \setminus (I_1 \cup \ldots \cup I_k)$.
  
  The proximity bounds for each $I_j$ follow directly from Lemmas~\ref{lem:easyproximity} and~\ref{lem:constronestepguarantee}. 
  
  By Lemma~\ref{lem:constrtime} the total running time of this procedure is $\tOh(k n)$. 
  
  It remains to show that $k = O(\log^2 n)$.
  To this end, for any non-empty set $U \subseteq [n]$ we define the potential $\phi(U) := \log_2(m_U) \cdot 2 \lceil \log_{4/3}(n) \rceil + \lceil \log_{4/3}(|J_U|) \rceil$ (where $m_U, J_U$ are defined in lines~\ref{step:singlestep_one} and~\ref{step:singlestep_two} of algorithm $\textsc{SingleStep}(U)$).
  Observe that $\phi(U) \in \mathbb{N}_0$, since $m_U$ is a power of 2.
  Also observe that $\phi(U) \le O(\log^2 n)$, since $m_U \le 2n$ and $|J_U| \le |U| \le n$. 
  We claim that $\phi(U_{j+1}) < \phi(U_j)$ holds for any $1 \le j < k$, see Claim~\ref{cla:phiU}.
  This yields $0 \le \phi(U_k) < \phi(U_{k-1}) < \ldots < \phi(U_1) \le O(\log^2 n)$, and thus $k \le O(\log^2 n)$. 
  It remains to prove the claim.
  
  \begin{claim} \label{cla:phiU}
    We have $\phi(U_{j+1}) < \phi(U_j)$ for any $1 \le j < k$. 
  \end{claim}
  \begin{proof}
  We use the notation $m_j := m_{U_j}$ and $J_j := J_{U_j}$ for any $j$.
  Since $U_{j+1} \subseteq U_j$ we have $m_{j+1} \le m_{j}$. We consider two cases: in one case $m_j$ decreases by a factor 2 and in the other case $m_j$ stays the same and $|J_j|$ decreases by a factor $3/4$. Thus, in both cases the potential decreases by at least 1. More details follow.
  
  \emph{Case 1: $m_{j+1} < m_{j}$.} Since $m_j, m_{j+1}$ are powers of 2 we then have $\log_2(m_{j+1}) \le \log_2(m_{j}) - 1$. Since $\log_2(m_U) \cdot 2 \lceil \log_{4/3}(n) \rceil$ is the dominant term in the definition of $\phi(U)$, we clearly have $\phi(U_{j+1}) < \phi(U_j)$. 
  More formally we can argue as follows. Since $|J_{j}| \ge 1$ and $|J_{j+1}| \le n$ we have
  \begin{align*}
    \phi(U_{j+1}) 
    &= \log_2(m_{j+1}) \cdot 2 \lceil \log_{4/3}(n) \rceil + \lceil \log_{4/3}(|J_{j+1}|) \rceil  \\
  &\le (\log_2(m_{j}) - 1) \cdot 2 \lceil \log_{4/3}(n) \rceil + \lceil \log_{4/3}(n) \rceil  \\
  &< \log_2(m_{j}) \cdot 2 \lceil \log_{4/3}(n) \rceil
  \le \log_2(m_{j}) \cdot 2 \lceil \log_{4/3}(n) \rceil + \lceil \log_{4/3}(|J_{j}|) \rceil 
  = \phi(U_j).
  \end{align*}
  
  \emph{Case 2: $m_{j+1} = m_{j}$.} In this case we observe that $J_{j+1} \subseteq J_{j} \setminus I_j$ since $U_{j+1} = U_j \setminus I_j$. 
  Using that $I_j \subseteq J_j$ and $|I_j| \ge |J_j|/4$ by Observation~\ref{obs:sizeI}, we obtain $|J_{j+1}| \le |J_j| - |I_j| \le \frac 34 |J_j|$. Thus, the potential drops:
  \begin{align*}
    \phi(U_{j+1}) 
    &= \log_2(m_{j+1}) \cdot 2 \lceil \log_{4/3}(n) \rceil + \lceil \log_{4/3}(|J_{j+1}|) \rceil  \\
  &\le \log_2(m_{j}) \cdot 2 \lceil \log_{4/3}(n) \rceil + \lceil \log_{4/3}(\tfrac 34 |J_j|) \rceil  \\
  &\le \log(m_{j}) \cdot 2 \lceil \log_{4/3}(n) \rceil + \lceil \log_{4/3}(|J_j|) \rceil - 1
  = \phi(U_j) - 1.
  \end{align*}
  In both cases we have $\phi(U_{j+1}) < \phi(U_j)$, which finishes the proof.
  \end{proof}
  Now that we proved Claim~\ref{cla:phiU} we finished the proof of Theorem~\ref{thm:partitioning}.
\end{proof}

\section{Algorithmic Ingredients}
\label{sec:algorithms}

In this section, after gathering some algorithmic ingredients in Sections~\ref{sec:prepmaxconv} and \ref{sec:greedyequal}, in Section~\ref{sec:knapsackalgo} we present an algorithm for 0-1-Knapsack that is given a partitioning and proximity bounds as in Theorem~\ref{thm:partitioning}. Combining the algorithm in Section~\ref{sec:knapsackalgo} with Theorem~\ref{thm:partitioning} then proves the main result for 0-1-Knapsack (as we will discuss in Section~\ref{sec:puttingtogether}).


\subsection{MaxPlusConv with a Concave Sequence} \label{sec:prepmaxconv}

A fundamental subroutine of all recent Knapsack algorithms is the $(\max,+)$-convolution operation, or MaxPlusConv  for short, which is defined as follows.

\begin{problem}[MaxPlusConv]
  Given two sequences $x = x[0..n]$ and $y = y[0..m]$ with entries in $\mathbb{Z} \cup \{-\infty\}$, compute the sequence $z = z[0..n+m]$ with $z[k] = \max\{ x[i] + y[k-i] : 0 \le i \le k \}$. Here out-of-bounds entries of $x$ and $y$ are interpreted as $-\infty$. We denote this operation by $z = x \star y$.
\end{problem}

MaxPlusConv is equivalent to the analogous problem MinPlusConv (with $\min$ replacing $\max$ and $\infty$ replacing $-\infty$ in the above definition).\footnote{This equivalence follows from negating all entries of $x$ and $y$, and negating the entries of the result $z$.}
MaxPlusConv and MinPlusConv are central problems in the area of fine-grained complexity theory~\cite{williams2018some}. On two sequences of length $n$ they can be solved naively in time $O(n^2)$, and a popular conjecture postulates that they cannot be solved in time $O(n^{2-\delta})$ for any constant $\delta > 0$~\cite{CyganMWW19,KunnemannPS17}. 

While in general MaxPlusConv is conjectured to require essentially quadratic time, on certain structured instances faster algorithms are known.
Here we make use of the special case where one of the sequences is concave, as has also been used, e.g., in~\cite{PolakRW21,CH2218a,AxiotisT19,KellererP04,
ChenLMZ23}.

\begin{definition}
  We say that a sequence $y \in (\mathbb{Z} \cup \{-\infty\})^m$ is \emph{concave} if there exists an offset $h \in \mathbb{N}$ and a length $\ell \in \mathbb{N}_0$ such that (1) the entries $y[0],y[h],\ldots,y[\ell \cdot h]$ are finite, (2) all other entries of $y$ are $-\infty$, (3) we have $y[i \cdot h] - y[i \cdot h - h] \ge y[i \cdot h + h] - y[i \cdot h]$ for all $1 \le i < \ell$. 
\end{definition}

\begin{lemma}[MaxPlusConv with Concave Sequence] \label{lem:maxplusconcave}
Given an arbitrary sequence $x \in (\mathbb{Z} \cup \{-\infty\})^n$ and a concave sequence $y \in (\mathbb{Z} \cup \{-\infty\})^m$, we can compute their MaxPlusConv $x \star y$ in time $O(n+m)$.
\end{lemma}
\begin{proof}
  This is a standard application of the SMAWK algorithm~\cite{AggarwalKMSW87}. For completeness we present the argument, following the presentation in~\cite{PolakRW21}. 
  For each remainder $r \in \{0,1,\ldots,h-1\}$ we separately compute the entries of $z = x \star y$ at indices that have remainder $r$ modulo $h$ (i.e., for a fixed $r$ our goal is to compute the entries $z[r], z[r + h], z[r + 2h], \ldots$ of $z$).
  We define the matrix $M \in (\mathbb{Z} \cup \{-\infty\})^{\lceil (n+m)/h \rceil \times \lceil n/h \rceil}$ with $M[i,j] = x[j \cdot h + r] + y[(i-j) h]$, where out-of-bounds entries of $x$ and~$y$ are interpreted as $-\infty$. We do not explicitly construct the matrix, but we can compute every entry of $M$ in constant time. Observe that the matrix $M$ is inverse-Monge, i.e.,
  \begin{align*}
    M[i,j] + M[i+1,j+1] &= x[j \cdot h + r] + y[(i-j) h] + x[(j+1) \cdot h + r] + y[(i-j) h]  \\
  &\ge x[j \cdot h + r] + y[(i-j) h + h] + x[(j+1) \cdot h + r] + y[(i-j) h - h]  \\
  &= M[i+1,j] + M[i,j+1].
  \end{align*}
  Therefore, the SMAWK algorithm~\cite{AggarwalKMSW87} is applicable and computes all row maximas of $M$ in total time $O((n+m)/h)$. Observe that the maximum of the $i$-th row of $M$ is equal to $z[i \cdot h + r]$. Running this algorithm for all $r \in \{0,1,\ldots,h-1\}$ yields all entries of $z$ and takes total time $O(n+m)$.
\end{proof}

\subsection{Greedy Algorithm for Equal Weights} \label{sec:greedyequal}

0-1-Knapsack is particularly easy to solve if all items have the same weight, because then the greedy algorithm yields an optimal solution. In this situation we can even solve the instance for every weight budget $0 \le W' \le W$ in total time $O(W)$, as shown by the following well-known lemma.

\begin{lemma}[0-1-Knapsack with Equal Weights] \label{lem:greedyequalweights}
  There is an algorithm \textsc{EqualWeights}$(w, p, W)$ that, given a 0-1-Knapsack instance $(w,p,W)$ of size $n$ with $w_1 = \ldots = w_n$ and $p_1 \ge \ldots \ge p_n$, in time $O(W)$ computes the sequence $y = y[0..W]$ with entries
  $$ y[W'] = \max\{ p^T x : w^T x = W',\, x \in \{0,1\}^n \} \quad\quad \text{for any $0 \le W' \le W$}. $$ 
\end{lemma}
\begin{proof}
The following simple greedy strategy is optimal in case of equal weights, and it clearly runs in time $O(W)$.


  \begin{algorithm}[H]
    \caption{\textsc{EqualWeights}$(w, p, W)$, assumes that all weights in $w$ are equal and $p$ is sorted}\label{alg:greedyequalweights}
    \begin{algorithmic}[1]
  \State Initialize $y = y[0..W]$ with $y[W'] = -\infty$ for all $0 \le W' \le W$
  \State Set $y[i \cdot w_1] := \sum_{j=1}^i p_j$ for $i=0,1,\ldots,\min\{n, \lfloor W/w_1 \rfloor \}$
  \State \Return $y$
    \end{algorithmic}
  \end{algorithm}
\end{proof}

\subsection{Algorithm for 0-1-Knapsack with given Proximity Bounds}
\label{sec:knapsackalgo}

We are now ready to present an algorithm for 0-1-Knapsack that is given a partitioning of the items and for each part a proximity bound, as in Theorem~\ref{thm:partitioning}. For more details see the following lemma.
This result for $k=1$ was proven in \cite[Lemma 2.2]{PolakRW21}. Our lemma is very similar to \cite[Lemma 6]{ChenLMZ23}, except that they partition the set of weights instead of the set of items.

\begin{lemma} \label{lem:knapsackgivenproximity}
  Suppose we are given a 0-1-Knapsack instance $(w,p,W)$ of size $n$ with distinct profit-to-weight ratios $p_1/w_1 > \ldots > p_n/w_n$, a partitioning $I_1 \cup \ldots \cup I_k = [n]$, and integers $\Delta_1 \le \ldots \le \Delta_k$. Denote the maximal prefix solution by $g$. Assume that there exists an optimal solution $x^*$ such that $\sum_{i \in I_j} w_i \cdot |g_i - x^*_i| \le \Delta_j$ holds for all $j$. 
  Then we can solve the 0-1-Knapsack instance in time $O(n + k \cdot \sum_{j=1}^k |\supp(w(I_j))| \cdot \Delta_j)$.
\end{lemma}

\begin{proof}
We start by presenting the intuition. 
We split each part $I_j$ into $I^+_j := \{ i \in I_j : g_i = 0 \}$ and $I^-_j := \{ i \in I_j : g_i = 1 \}$. 
The optimal solution $x^*$ uses weight at most $\Delta_j$ among the items in $I^+_j$, and thus we can restrict our attention to weights $0 \le W' \le \Delta_j$ in $I^+_j$. Similarly, $g - x^*$ uses weight at most $\Delta_j$ among the items in $I^-_j$, and thus we can restrict our attention to weights $0 \le W' \le \Delta_j$ in~$I^-_j$. For simplicity, in this overview we focus on the set $I^+_j$.
We further split each set $I^+_j$ according to item weights; this splits $I^+_j$ into sets $I^+_{j,1},\ldots,I^+_{j,\ell_j}$ where $\ell_j := |\supp(w(I^+_j))| \le |\supp(w(I_j))|$. All items in $I^+_{j,\ell}$ have the same weight, so the greedy algorithm from Lemma~\ref{lem:greedyequalweights} computes the optimal profits for each weight $0 \le W' \le \Delta_j$ in time $O(\Delta_j)$. 
Note that the total time to run the greedy algorithm on each set $I^+_{j,\ell}$ with weight bound $\Delta_j$ is $O(\sum_{j = 1}^k \sum_{\ell = 1}^{\ell_j} \Delta_j) \le  O(\sum_{j=1}^k |\supp(w(I_j))| \cdot \Delta_j)$, plus time $O(n)$ for initialization.

It remains to combine the results of the calls to the greedy algorithm. We implement the combination operation by the MaxPlusConv algorithm for concave sequences from Lemma~\ref{lem:maxplusconcave}, and we show that the right ordering of combination operations increases the running time by only a factor $O(k)$. Hence, in total the running time is $O(n + k \cdot \sum_{j=1}^k |\supp(w(I_j))| \cdot \Delta_j)$.
In what follows we make this intuition precise.

\medskip
The following pseudocode presents our algorithm for 0-1-Knapsack.

  \begin{algorithm}[H]
    \caption{\textsc{0-1-Knapsack}$(w, p, W)$, given an instance $(w,p,W)$ of size $n$ with distinct profit-to-weight ratios $p_1/w_1 > \ldots > p_n/w_n$, a partitioning $I_1 \cup \ldots \cup I_k = [n]$ and integers $\Delta_1 \le \ldots \le \Delta_k$}\label{alg:zerooneknapsack}
    \begin{algorithmic}[1]
  \State Compute the maximal prefix solution $g$ \label{step:knapsack_one}
  \State Initialize sequences $z^+ = (0)$ and $z^- = (0)$ \label{step:knapsack_two}
  \For{$j=1,\ldots,k$} \label{step:knapsack_loopj}
    \State Compute the set $\supp(w(I_j))$ and write $\supp(w(I_j)) = \{\bar w_{j,1},\ldots,\bar w_{j,\ell_j}\}$ \label{step:knapsack_four}
    \State Compute the sets $I_{j,\ell}^+ = \{ i \in I_j : w_i = \bar w_{j,\ell},\, g_i = 0 \}$ and $I_{j,\ell}^- = \{ i \in I_j : w_i = \bar w_{j,\ell},\, g_i = 1 \}$ \\ \hspace{0.8cm} in sorted order for all $1 \le \ell \le \ell_j$.  \label{step:knapsack_six}
    \For{$\ell = 1,\ldots,\ell_j$} \label{step:knapsack_loopl}
      \State Compute $y^+ := \textsc{EqualWeights}((w_i : i \in I_{j,\ell}^+), (p_i : i \in I_{j,\ell}^+), \Delta_j)$ using Lemma~\ref{lem:greedyequalweights} \label{step:knapsack_eight}
      \State Compute $z^+ := z^+ \star y^+$ using Lemma~\ref{lem:maxplusconcave}
      \State Truncate $z^+$ to indices $0,\ldots,\sum_{j' \le j} \Delta_{j'}$
      \State Compute $y^- := \textsc{EqualWeights}((w_i : i \in I_{j,\ell}^-), (- p_i : i \in I_{j,\ell}^-), \Delta_j)$ using Lemma~\ref{lem:greedyequalweights} \label{step:knapsack_negation}
      \State Compute $z^- := z^- \star y^-$ using Lemma~\ref{lem:maxplusconcave}
      \State Truncate $z^-$ to indices $0,\ldots,\sum_{j' \le j} \Delta_{j'}$ \label{step:knapsack_endpartone}
    \EndFor
  \EndFor
  \State $s^+[0] := z^+[0]$ \label{step:knapsack_startparttwo}
  \For{$i=1,\ldots,\sum_{j \le k} \Delta_j$} $s^+[i] := \max\{ s^+[i-1], z^+[i] \}$ \EndFor
  \State \Return $p^T g + \max\{ s^+[\min\{b + W - w^T g, \sum_{j \le k} \Delta_j\}] + z^-[b] : 0 \le b \le \sum_{j \le k} \Delta_j \}$ \label{step:knapsack_endparttwo}
    \end{algorithmic}
  \end{algorithm}

In what follows we analyze correctness and running time of this algorithm.

\paragraph{Correctness}
We start by analyzing lines~\ref{step:knapsack_one}-\ref{step:knapsack_endpartone} of the algorithm, as formalized in the following claim.
Here we use the following notation:
Let $I^+ := \{i \in [n] : g_i = 0\}$ and $I^- := [n] \setminus I^+$. For the optimal solution $x^*$, denote by $x^+$ the restriction of $x^*$ to $I^+$, 
i.e., for any $i \in I^+$ we have $x^+_i = x^*_i$ and for any $i \in [n] \setminus I^+$ we have $x^+_i = 0$. Note that $x^+$ are the items that we add to the maximal prefix solution $g$ to obtain $x^*$ from $g$.
Also denote by $x^-$ the restriction of $g - x^*$ to~$I^-$. Note that $x^-$ are the items that we remove from the maximal prefix solution $g$ to obtain $x^*$ from $g$.

\begin{claim} \label{cla:algoknapsackcorrect}
  At the beginning of line~\ref{step:knapsack_startparttwo}, for any $0 \le W' \le \sum_{j \le k} \Delta_j$ we have:
  \begin{enumerate}[label=(\roman*)]
  \item $z^+[W']$ is the profit of some set of items in $I^+$ with total weight $W'$, or it is $-\infty$,
  \item $z^-[W']$ is the negated profit of some set of items in $I^-$ with total weight $W'$, or it is $-\infty$,
  \item $z^+[w^T x^+] \ge p^T x^+$, in particular $w^T x^+$ is a valid index of $z^+$,
  \item $z^-[w^T x^-] \ge - p^T x^-$, in particular $w^T x^-$ is a valid index of $z^-$.
  \end{enumerate}
\end{claim}
\begin{proof}
We start by setting up notation.
We call the $j$-th iteration of the loop in line~\ref{step:knapsack_loopj} \emph{phase $j$}, and we call the $\ell$-th iteration of the loop in line~\ref{step:knapsack_loopl} during phase $j$ \emph{iteration $(j,\ell)$}. 
For any $j \in [k], \ell \in [\ell_j]$ we define
the sets $I^+_j := I_j \cap I^+$, $I^+_{\le j} = I^+_1 \cup \ldots \cup I^+_j$, and 
$I^+_{\le (j,\ell)} = I^+_1 \cup \ldots \cup I^+_{j-1} \cup I^+_{j,1} \cup \ldots \cup I^+_{j,\ell}$. Note that $I^+_{\le (j,\ell)}$ are the items processed in $z^+$ at the end of iteration $(j,\ell)$.  

We first claim that at the end of iteration $(j,\ell)$ the number $z^+[W']$ is the profit of some set of items in $I^+_{\le (j,\ell)}$ of weight equal to $W'$, or it is $-\infty$, for every $0 \le W' \le \sum_{j' \le j} \Delta_{j'}$. By focussing on the last iteration $(k,\ell_k)$, this implies bullet (i).
The claim can be shown by induction, by using three ingredients: (1) We can inductively assume that the claim holds before iteration $(j,\ell)$. (2) By correctness of \textsc{EqualWeights}, $y^+[W']$ is the maximum profit of any subset of the items $I^+_{j,\ell}$ with total weight equal to $W'$. (3) $z^+ \star y^+$ satisfies $(z^+ \star y^+)[W'] = \max_{W_1+W_2 = W'} z^+[W_1] + y^+[W_2]$, which is equal to the maximum profit sum of any two solutions in $z^+$ and $y^+$ with weights summing to~$W'$. This proves the claim.

The proof of bullet (ii) is symmetric, up to the sign change of profits due to the negation in line~\ref{step:knapsack_negation}.

We introduce some additional notation: For the fixed optimal solution $x^*$, denote by $x^*[J]$ the restriction of $x^*$ to the items in $J$, i.e., $x^*[J]$ is a vector in $\{0,1\}^n$ with $x^*[J]_i = 1$ if $x^*_i = 1$ and $i \in J$, and $x^*[J]_i = 0$ otherwise.
Denote by $z^{+,j}$ the sequence $z^+$ at the end of phase $j$.

We claim that $z^{+,j}[w^T x^*[I^+_{\le j}]] \ge p^T x^*[I^+_{\le j}]$. By focussing on the last phase $j = k$ this implies bullet (iii), since $x^*[I^+_{\le k}] = x^*[I^+] = x^+$. The claim can be shown by induction using the proximity bounds guaranteed by the lemma statement: We can inductively assume that $z^{+,j-1}[w^T x^*[I^+_{\le j-1}]] \ge p^T x^*[I^+_{\le j-1}]$, in particular $w^T x^*[I^+_{\le j-1}]$ is a valid index in $z^{+,j-1}$, so $w^T x^*[I^+_{\le j-1}] \le \sum_{j' \le j-1} \Delta_{j'}$. By the proximity bound guaranteed by the lemma statement we have $w^T x^*[I^+_{j}] \le \Delta_j$, so $w^T x^*[I^+_{\le j}] = w^T x^*[I^+_{\le j-1}] + w^T x^*[I^+_{j}] \le \sum_{j' \le j} \Delta_{j'}$ is a valid index in $z^{+,j}$. Now observe that in iteration $(j,\ell)$ we have $y^+[w^T x^*[I^+_{j,\ell}]] \ge p^T x^*[I^+_{j,\ell}]$, since \textsc{EqualWeights} computes optimal solutions. Using the summand $y^+[w^T x^*[I^+_{j,\ell}]]$ in the MaxPlusConv in iteration $(j,\ell)$, we obtain that 
$$ z^{+,j}[w^T x^*[I^+_{\le j}]] \ge z^{+,j-1}[w^T x^*[I^+_{\le j-1}]] + \sum_{\ell=1}^{\ell_j} p^T x^*[I^+_{j,\ell}] \ge p^T x^*[I^+_{\le j-1}] + \sum_{\ell=1}^{\ell_j} p^T x^*[I^+_{j,\ell}] = p^T x^*[I^+_{\le j}], $$ 
which proves the claim.

The proof of bullet (iv) is symmetric, up to the sign change of profits in line~\ref{step:knapsack_negation}.
\end{proof}

We use the above claim to arrive at the following characterization of the optimal profit $p^T x^*$ in terms of the sequences $z^+$ and $z^-$.

\begin{claim} \label{cla:algoknapsackcorrecttwo}
At the beginning of line~\ref{step:knapsack_startparttwo}, the optimal profit $p^T x^*$ satisfies
\begin{align} \label{eq:algokanpsacaksd}
  p^T x^* = p^T g + \max\{ z^+[a] + z^-[b] : 0 \le a,b \le \sum_{j \le k} \Delta_j,\, a-b \le W - w^T g \}.
\end{align}
\end{claim}
\begin{proof}
For one direction, set $a := w^T x^+$ and $b := w^T x^-$ and use Claim~\ref{cla:algoknapsackcorrect} to obtain
$$ p^T g + z^+[a] + z^-[b] \ge p^T g + p^T x^+ - p^T x^- = p^T x^*, $$
where in last equality we used that $p^T x^+$ is the profit of all items in $I^+$ selected by $x^*$, and $p^T x^-$ is the profit of items removed from $g$ to obtain $x^*$, so $p^T g - p^T x^-$ is the profit of all items in $I^-$ selected by~$x^*$.
Claim~\ref{cla:algoknapsackcorrect} promises that $a,b$ are valid indices of $z^+,z^-$, and thus $0 \le a,b \le \sum_{j \le k} \Delta_j$. Moreover, we have $W \ge w^T x^* = w^T x^+ + w^T g - w^T x^- = a + w^T g - b$, so the constraint $a-b \le W - w^T g$ is satisfied.

For the other direction, for any $a,b$ we use that $z^+[a]$ is the profit of some set of items in $I^+$ with total weight $a$ (or $-\infty$) and $z^-[b]$ is the negated profit of some set of items in $I^-$ with total weight $b$ (or $-\infty$), so $p^T g + z^-[b]$ is the profit of some set of items in $I^-$ with total weight $w^T g - b$ (or $-\infty$). Thus, $z^+[a] + p^t g + z^-[b]$ is the profit of some set of items in $I^+ \cup I^- = [n]$ with total weight $a + w^T g - b$ (or $-\infty$). Hence, for all $a,b$ with $a-b \le W - w^T g$ the value $z^+[a] + p^t g + z^-[b]$ is the profit of some feasible solution (or $-\infty$), so we have $z^+[a] + p^t g + z^-[b] \le p^T x^*$. 
\end{proof}

It remains to argue that lines~\ref{step:knapsack_startparttwo}-\ref{step:knapsack_endparttwo} of the algorithm compute the right hand side of equation~(\ref{eq:algokanpsacaksd}).
To this end, note that the sequence $s^+$ stores the maximum of every prefix of $z^+$, i.e.,
$$ s^+[a] := \max\{ z^+[a'] : 0 \le a' \le a \}, $$
for every $0 \le a \le \sum_{j \le k} \Delta_j$.
Note that in the right hand side of equation (\ref{eq:algokanpsacaksd}) the maximum allowed value for $a$ is $\min\{b + W - w^T g, \sum_{j \le k} \Delta_j\}$. We can replace the maximization over $z^+[a]$ by $s^+$ evaluated at the maximum allowed value for $a$.  This yields
$$ p^T x^* = p^T g + \max\{ s^+[\min\{b + W - w^T g, \sum_{j \le k} \Delta_j\}] + z^-[b] : 0 \le b \le \sum_{j \le k} \Delta_j \}. $$
This shows that in line~\ref{step:knapsack_endparttwo} the algorithm computes the optimal profit $p^T x^*$, proving correctness.

\paragraph{Running time}
Lines~\ref{step:knapsack_one}-\ref{step:knapsack_two} run in time $O(n)$. Lines~\ref{step:knapsack_four}-\ref{step:knapsack_six} can be implemented by a scan over the set $I_j$, and thus run in total time $O(|I_1| + \ldots + |I_k|) = O(n)$. Lines~\ref{step:knapsack_startparttwo}-\ref{step:knapsack_endparttwo} run in time $O(\sum_{j \le k} \Delta_j)$
(since after precomputing $w^T g$ and $\sum_{j \le k} \Delta_j$ each evaluation in line~\ref{step:knapsack_endparttwo} takes constant time).
Note that in phase $j$ the length of $z^+, z^-$ is equal to $\sum_{j' \le j} \Delta_{j'} \le j \cdot \Delta_j$. Lines~\ref{step:knapsack_eight}-\ref{step:knapsack_endpartone} thus all run in time $O(j \cdot \Delta_j)$. The loop in line~\ref{step:knapsack_loopl} has $\ell_j = |\supp(w(I_j))|$ iterations. 
Hence, the total running time is 
$$ O\big(n + \sum_{j \le k} |\supp(w(I_j))| \cdot j \cdot \Delta_j\big) \le O\big(n + k \cdot \sum_{j \le k} |\supp(w(I_j))| \cdot \Delta_j\big). $$
%
This finishes the proof of Lemma~\ref{lem:knapsackgivenproximity}.
%
\end{proof}

\section{Reduction to 0-1-Knapsack}
\label{sec:reduction}

In this section we use the classic proximity bound~\cite{EisenbrandW20} to reduce Bounded Knapsack to 0-1-Knapsack on $O(\wmax^2)$ items. The techniques used in this reduction are similar to prior algorithms for Bounded Knapsack, e.g.,~\cite{EisenbrandW20,PolakRW21}, but the formulation as a reduction is new.

\begin{restatable}[Reduction]{lemma}{lemreduction}
\label{lem:reduction}
  If 0-1-Knapsack on $O(\wmax^2)$ items with distinct profit-to-weight ratios can be solved in time $\tOh(\wmax^2)$, then Bounded Knapsack can be solved in time $\tOh(n + \wmax^2)$.
\end{restatable}



\begin{proof}
We start with a proof overview.
We first present a reduction from 0-1-Knapsack to 0-1-Knapsack with $O(\wmax^2)$ items.
To this end, consider a given 0-1-Knapsack instance $(w,p,W)$ and its maximal prefix solution $g$.
For any weight $\hat w \in [\wmax]$, among the items of weight $\hat w$ that are picked by~$g$, we greedily pick all but the $2 \wmax$ least profitable items. 
Similarly, among the items of weight~$\hat w$ that are not picked by $g$, we remove all but the $2 \wmax$ most profitable items. 
The classic proximity bound~\cite{EisenbrandW20} shows that there exists an optimal solution that is consistent with these choices.
The remaining instance has at most $4\wmax$ items of each weight in $[\wmax]$, so we have a reduction to $O(\wmax)$ items. 

A Bounded Knapsack instance can be converted to a 0-1-Knapsack instance by replacing each item $i$ with multiplicity $u_i$ by $u_i$ copies of item $i$. Therefore, the same reduction as above also works starting from Bounded Knapsack. We show that this reduction can be implemented efficiently without explicitly blowing up the original Bounded Knapsack instance to a 0-1-Knapsack instance.
The result of the reduction then is a Bounded Knapsack instance with total multiplicity of all items $O(\wmax^2)$. Now we can afford to explicitly list all copies of items to arrive at 0-1-Knapsack.

Finally, we scale up all profits by a large factor and then add small noise terms. This maintains the optimal solution and makes all profit-to-weight ratios distinct, showing the lemma.

%
In what follows we give the details of this reduction.

\paragraph{Reduction from 0-1-Knapsack}
We first give a reduction from 0-1-Knapsack to 0-1-Knapsack with $O(\wmax^2)$ items. 
Consider a 0-1-Knapsack instance $(w,p,W)$ of size $n$, and sort it by non-increasing profit-to-weight ratio $p_1/w_1 \ge \ldots \ge p_n/w_n$.
We iterate over all weights $\hat w \in [\wmax]$, so in what follows fix the weight $\hat w$.
Denote by $i_1 < \ldots < i_{\ell}$ all indices of items of weight $\hat w$. Note that $p_{i_1} \ge \ldots \ge p_{i_{\ell}}$. 
The maximal prefix solution $g$ picks a prefix of these items, i.e., for some index $t$ we have $g_{i_j} = 1$ for all $1 \le j \le t$ and $g_{i_j} = 0$ for all $t < j \le \ell$. 
Among the items of weight $\hat w$ not picked by $g$ we want to remove all but the $2 \wmax$ most profitable items, i.e., we want to remove the items $i_{t+2\wmax+1}, i_{t+2\wmax+2},\ldots, i_\ell$. Formally, we add all items $i_{t+2\wmax+1}, i_{t+2\wmax+2},\ldots, i_\ell$ to an initially empty set $I^{(0)}$ that stores all removed items.
Similarly, among the items of weight $\hat w$ picked by $g$ we want to pick all but the $2 \wmax$ least profitable items, i.e., we want to pick the items $i_1,\ldots,i_{t-2\wmax}$. Formally, we add all items $i_1,\ldots,i_{t-2\wmax}$ to an initially empty set $I^{(1)}$ that stores all picked items. 
After performing this process for each weight $\hat w \in [\wmax]$, we arrive at some set of items $I^{(0)}$ that are to be removed and some set of items $I^{(1)}$ that are to be picked. We remember the total profit of the picked  items $P := \sum_{i \in I^{(1)}} p_i$ and reduce the weight bound accordingly to $\bar W := W - \sum_{i \in I^{(1)}} w_i$. Then we remove the coordinates in $I^{(0)} \cup I^{(1)}$ from $w$ and $p$. Call the resulting 0-1-Knapsack instance $(\bar w, \bar p, \bar W)$, and denote its size by $\bar n := n - |I^{(0)}| - |I^{(1)}|$. Note that at most $4 \wmax$ items of any weight in $[\wmax]$ remain, so we have $\bar n \le 4 \wmax^2$. 

We claim that the choices we have made are without loss of generality, i.e., there is an optimal solution $x^* \in \{0,1\}^n$ of the original 0-1-Knapsack instance $(w,p,W)$ that selects the items in $I^{(1)}$ and does not select the items in $I^{(0)}$. This implies,
$$ \max\{ \bar p^T x : \bar w^T x \le \bar W, x \in \{0,1\}^{\bar n} \} + P = \max\{ p^T x : w^T x \le W, x \in \{0,1\}^n \}. $$
To prove the claim, consider all optimal solutions $x^*$ that minimize $\sum_{i=1}^n |g_i - x^*_i|$, and among all such solutions pick an optimal solution $x^*$ minimizing $\sum_{i \in I^{(0)}} x^*_i - \sum_{i \in I^{(1)}} x^*_i$.
Suppose for the sake of contradiction that $x^*$ selects one of the items $i \in I^{(0)}$ that we have removed. We consider two cases.

\emph{Case 1: There is an item $j \in [n] \setminus (I^{(0)} \cup I^{(1)})$ of weight $w_j = w_i$ that is not selected by $x^*$.} Then we can exchange item $i$ by item~$j$ in $x^*$, i.e., we set $x^*_i = 0$ and $x^*_j = 1$. Since both items have the same weight, this exchange maintains the weight $w^T x^*$. Since we removed items of weight $w_i$ with the smallest profits, we have $p_j \ge p_i$. Thus, after this exchange $x^*$ is still an optimal solution. Since this change does not increase $\sum_{i=1}^n |g_i - x^*_i|$ and decreases $\sum_{i \in I^{(0)}} x^*_i - \sum_{i \in I^{(1)}} x^*_i$, we obtain a contradiction to the choice of $x^*$. 

\emph{Case 2: All items $j \in [n] \setminus (I^{(0)} \cup I^{(1)})$ of weight $w_j = w_i$ are selected by $x^*$.}
Let $i_1 < \ldots < i_\ell$ be all indices of items with weight $w_i$, and let $t$ be such that $g$ picks items $i_1,\ldots,i_t$ but not $i_{t+1},\ldots,i_\ell$. Since $i_{t+1},\ldots,i_{t+2\wmax} \in [n] \setminus (I^{(0)} \cup I^{(1)})$, it follows that $x^*$ selects all items  $i_{t+1},\ldots,i_{t+2\wmax}$, and thus $\sum_{i=1}^n |g_i - x^*_i| \ge 2\wmax$. 
However, this contradicts the following classic proximity bound.
\begin{lemma}[\cite{EisenbrandW20,PolakRW21}]
  We have $\sum_{i=1}^n |g_i - x^*_i| \le 2\wmax - 1$.
\end{lemma}
\begin{proof}
For completeness, we repeat the proof of this proximity bound as presented in~\cite{PolakRW21}. 
Note that $x^*$ is maximal in the sense that no item can be added to it, and $g$ is maximal in the sense that it selects a maximal prefix of the items. It follows that either $w^T x^* = w^T g = \sum_{i \in [n]} w_i$ or $w^T x^*, w^T g \in (W-\wmax, W]$. Both yield 
\begin{align} \label{eq:difference}
  -\wmax < w^T (x^* - g) < \wmax.
\end{align}
Now consider the following process. Start with the vector $x^* - g$. We will move its entries to 0, while maintaining the inequalities (\ref{eq:difference}). That is, in each step of the process, if the current sum $w^T (x^* - g)$ is positive we reduce an arbitrary positive entry of $x^* - g$ by 1, otherwise we increase an arbitrary negative entry by~1. During this process in no two steps we can have the same sum, as otherwise we could apply to~$x^*$ the additions and removals performed between these two steps, obtaining another solution~$x'$. Because at both steps we had the same sum, the weight does not change, i.e., $w^T x' = w^T x^*$, so $x'$ is feasible. Since every item selected by $g$ but not $x^*$ has no lower profit than any item selected by $x^*$ but not $g$, the additions performed to obtain $x'$ from $x^*$ have a higher average profit-to-weight ratio than the removals. Since both have the same total weight, the total profit does not decrease, i.e., $p^T x' \ge p^T x^*$, so $x'$ is also an optimal solution. 
Thus, the new solution $x'$ is closer to $g$ and still optimal, contradicting the choice of $x^*$. 

Hence, the number of steps of this process is at most $2\wmax-1$. Observing that the number of steps is equal to $\sum_{i \in [n]} |x^*_i - g_i|$ finishes the proof.
\end{proof}

In both cases we arrived at a contradiction, so $x^*$ cannot select any item in $I^{(0)}$. The proof that $x^*$ selects all items in $I^{(1)}$ is symmetric. This finishes the reduction from 0-1-Knapsack to 0-1-Knapsack with $O(\wmax^2)$ items.

\paragraph{Reduction from Bounded Knapsack}
Note that there is a trivial reduction from a Bounded Knapsack instance $(w,p,u,W)$ to an equivalent 0-1-Knapsack instance $(w',p',W)$ that simply replaces every item $w_i,p_w,u_i$ by $u_i$ copies of the item $w_i,p_i$. Combining this reduction with the reduction from the last paragraph (from 0-1-Knapsack to 0-1-Knapsack with $O(\wmax^2)$ items) yields a reduction from Bounded Knapsack to 0-1-Knapsack with $O(\wmax^2)$ items. However, a naive implementation of this reduction would require time $\Omega(\sum_i u_i)$, which is not efficient enough. Therefore, we now describe how to implement the same reduction in an implicit way directly on Bounded Knapsack, running in time $\tOh(n + \wmax^2)$. This is essentially straightforward but we describe the details for completeness. We do not need to reprove correctness, as we merely present a different implementation of the same reduction.

Consider a Bounded Knapsack instance $(w,p,u,W)$ of size $n$, and sort it by non-increasing profit-to-weight ratio $p_1/w_1 \ge \ldots \ge p_n/w_n$. Initialize $u^{(0)} = u^{(1)} = (0,\ldots,0) \in \mathbb{N}_0^n$.
We iterate over all weights $\hat w \in [\wmax]$, so in what follows fix the weight $\hat w$.
Denote by $i_1 < \ldots < i_{\ell}$ all indices $i$ such that $w_i = \hat w$. Note that $p_{i_1} \ge \ldots \ge p_{i_{\ell}}$. 
The maximal prefix solution $g$ picks a prefix of these items, i.e., for some index $t$ we have $g_{i_j} = u_{i_j}$ for all $1 \le j < t$ and $g_{i_j} = 0$ for all $t < j \le \ell$. 

Among the items of weight $\hat w$ not picked by $g$ we want to remove all but the $2\wmax$ most profitable items, as described in more detail in what follows. If $u_{i_t} - g_{i_t} + \sum_{j=t+1}^{\ell} u_{i_j} \le 2\wmax$ there are at most $2\wmax$ items of weight $\hat w$ not picked by $g$, so there is nothing to do.
Otherwise, while $\ell > t$ and $u_{i_t} - g_{i_t} + \sum_{j=t+1}^{\ell-1} u_{i_j} \ge 2\wmax$, we set $u^{(0)}_{i_\ell} := u_{i_\ell}$, thereby marking all copies of the item $i_\ell$ for removal, and we set $\ell := \ell-1$. 
This process stops with $\ell = t$ or $u_{i_t} - g_{i_t} + \sum_{j=t+1}^{\ell-1} u_{i_j} < 2\wmax$, and furthermore $u_{i_t} - g_{i_t} + \sum_{j=t+1}^{\ell} u_{i_j} \ge 2\wmax$. We then set $u^{(0)}_{i_\ell} := u_{i_t} - g_{i_t} + (\sum_{j=t+1}^{\ell} u_{i_j}) - 2\wmax$, marking this many copies of item $i_\ell$ for removal, which leaves us with exactly $2\wmax$ items of weight $\hat w$ that are not picked by $g$.

Similarly, among the items of weight $\hat w$ picked by $g$ we want to pick all but the $2\wmax$ least profitable items, as described in more detail in what follows. If $\sum_{j=1}^{t} g_{i_j} \le 2\wmax$ there are at most $2\wmax$ items of weight $\hat w$ picked by $g$, so there is nothing to do.
Otherwise, set $a := 1$. While $a < t$ and $\sum_{j=a+1}^{t} g_{i_j} \ge 2\wmax$, we set $u^{(1)}_{i_a} := u_{i_a}$, thereby marking all copies of the item $i_a$ for picking, and we set $a := a+1$. 
This process stops with $a = t$ or $\sum_{j=a+1}^{t} g_{i_j} < 2\wmax$, and furthermore $\sum_{j=a}^{t} g_{i_j} \ge 2\wmax$. We then set $u^{(1)}_{i_a} := (\sum_{j=a}^{t} g_{i_j}) - 2\wmax$, marking this many copies of item $i_a$ for picking, which leaves us with exactly $2\wmax$ items of weight $\hat w$ that are picked by $g$.

After performing this process for each weight $\hat w \in [\wmax]$, we arrive at some multiplicity vectors $u^{(0)}$ marking items to be removed and $u^{(1)}$ marking items to be picked. We remember the total profit of the picked items $P := \sum_{i \in [n]} u^{(1)}_i \cdot p_i$ and reduce the weight bound accordingly to $\bar W := W - \sum_{i \in [n]} u^{(1)}_i \cdot w_i$. The we replace $u_i$ by $u'_i := u_i - u^{(0)}_i - u^{(1)}_i$ for all $i$. The remaining instance has total multiplicity $\sum_{i \in [n]} u'_i \le 4 \wmax^2$, since we have at most $4\wmax$ items left for each weight in $[\wmax]$. 
Therefore, we can explicitly write down all copies of all these items to form an equivalent 0-1-Knapsack instance $(\bar w, \bar p, \bar W)$ of size $\bar n \le 4 \wmax^2$. 

It remains to observe that this reduction computes exactly the same 0-1-Knapsack instance as the combination of the trivial reduction from Bounded Knapsack to 0-1-Knapsack and the reduction from 0-1-Knapsack to 0-1-Knapsack with $O(\wmax^2)$ items from the last paragraph.
This implies that we get an analogous correctness guarantee:
$$ \max\{ \bar p^T x : \bar w^T x \le \bar W, x \in \{0,1\}^{\bar n} \} + P = \max\{ p^T x : w^T x \le W, 0 \le x \le u, x \in \mathbb{N}_0^n \}. $$

Finally, note that the running time of this reduction is $\tOh(n + \wmax^2)$, as the main process runs in time $\tOh(n)$ and writing down the 0-1-Knapsack instance with $O(\wmax^2)$ items takes time $O(\wmax^2)$.

\paragraph{Distinct Profit-to-Weight Ratios}
It remains to make the profit-to-weight ratios distinct. After sorting we can assume that $\bar p_1/\bar w_1 \ge \ldots \ge \bar p_{\bar n}/\bar w_{\bar n}$. We replace $\bar p$ by $\bar {\bar p} \in \mathbb{N}^{\bar n}$ with $\bar{\bar p}_i = M \cdot \bar p_i + (\bar n-i) \cdot \bar w_i$ for a sufficiently large integer $M$. 
Note that for any $i < j$ we have
$$ \bar{\bar p}_i / \bar w_i = M \bar p_i / \bar w_i + (\bar n - i) > M \bar p_j / \bar w_j + (\bar n - j) = \bar{\bar p}_j / \bar w_j, $$
where we used $\bar p_i / \bar w_i \ge \bar p_j / \bar w_j$ and $i < j$. Therefore, all new profit-to-weight ratios are distinct.


It remains to show that any optimal solution $x \in \{0,1\}^{\bar n}$ for the 0-1-Knapsack instance $(\bar w, \bar{\bar p}, \bar W)$ is also an optimal solution for the instance $(\bar w, \bar p, \bar W)$. To see this, note that for any $x \in \{0,1\}^{\bar n}$ we have 
$$ \bar{\bar p}^T x = M \cdot \bar p^T x + \sum_{i \in [\bar n], x_i = 1} (\bar n - i) \cdot \bar w_i. $$
Since profits are integers and $M$ is sufficiently large, any increase in $\bar p^T x$ weighs more than any change in the second summand. Therefore, a maximizer for $\bar{\bar p}^T x$ is also a maximizer for $\bar p^T x$. (It suffices to set $M > \bar n^2 \wmax$ for this argument.) 

\medskip
Note that the reduction runs in time $\tOh(n + \wmax^2)$. Hence, if 0-1-Knapsack on $O(\wmax^2)$ items with distinct profit-to-weight ratios can be solved in time $\tOh(\wmax^2)$ then Bounded Knapsack can be solved in time $\tOh(n + \wmax^2)$.
This finishes the proof of Lemma~\ref{lem:reduction}.
\end{proof}

\section{Putting Everything Together}
\label{sec:puttingtogether}

In this section we put together the tools that we gathered in the previous sections to prove our main result.
We start by showing that 0-1-Knapsack can be solved in time $\tOh(n + \wmax^2)$.

\begin{theorem} \label{thm:zerooneknapsack}
  0-1-Knapsack on $n$ items with distinct profit-to-weight ratios can be solved in time $\tOh(n + \wmax^2)$.
\end{theorem}
\begin{proof}
  By combining Theorem~\ref{thm:partitioning} and Lemma~\ref{lem:knapsackgivenproximity}, we can solve 0-1-Knapsack in time 
  $$ \tOh\big(n + k \cdot \sum_{j \le k} |\supp(w(I_j))| \cdot \Delta_j\big) \le \tOh(n + k^2 \log^2(n) \wmax^2) \le \tOh(n + \wmax^2), $$
  where we used $|\supp(w(I_j))| \cdot \Delta_j \le O(\log^3(n) \wmax^2)$ and $k \le O(\log^2 n)$ from Theorem~\ref{thm:partitioning}.
\end{proof}

Now we can prove our main theorem, restated here for convenience.

\thmmain*
\begin{proof}
  This is immediate from combining Theorem~\ref{thm:zerooneknapsack} and Lemma~\ref{lem:reduction}.
\end{proof}

Alternatively, we can parametrize by the maximum profit $\pmax$.

\thmmainprofit*
\begin{proof}
  This result follows from Theorem~\ref{thm:main} by a simple reduction described in \cite[Section 4]{PolakRW21}.
\end{proof}

\section{Open Problems}

In this paper we presented an algorithm that solves Bounded Knapsack in time $\tOh(n + \wmax^2)$, which matches a conditional lower bound. This raises several open problems.

\begin{itemize}
\item Our algorithm has an impractical number of logarithmic factors. Indeed, the $\tOh$ in our running time hides at least a factor $\log^7 n$, where a factor $k^2 = \log^4 n$ comes from the number of parts in our partitioning, and a factor $\log^3 n$ comes from the additive combinatorics bounds. The independent work by Jin~\cite{JinArxiv23} has a factor $\log^4 \wmax$. Can these factors be improved? 
\item Our algorithm has impractical constant factors. Indeed, in Theorem~\ref{thm:partitioning} we have a factor 300000000. As this is only slightly larger than the constant $\approx$17000000 in Theorem~\ref{thm:BWtwo}, the main bottleneck are the constant factors in~\cite{BringmannW21}. Reducing the constant to a practical value might require a significant change in the proof approach of~\cite{BringmannW21}.
\item MaxPlusConv is conjectured to require essentially quadratic time in the worst case, but there are lower order improvements to time $n^2/2^{\Omega(\sqrt{\log n})}$ (by combining a reduction to All-Pairs Shortest Paths~\cite{BremnerCDEHILPT14} with an algorithm for the latter~\cite{Williams18}). Are similar improvements possible for Bounded Knapsack or 0-1-Knapsack? That is, can these problems be solved in time $\tOh(n + \wmax^2 / 2^{\Omega(\sqrt{\log \wmax})})$?
\item Can Bounded Knapsack or 0-1-Knapsack be solved in time $\tOh(n \wmax)$? Our approach does not seem able to give such a running time bound.
\item Can Bounded Knapsack or 0-1-Knapsack be solved in time $\tOh(n + (\wmax + \pmax)^{2-\delta})$ for some $\delta > 0$? Time $\tOh(n + (\wmax + \pmax)^{1.5})$ was recently shown for Unbounded Knapsack~\cite{BringmannC22}.
\item For Subset Sum the same parametrization by the number of items $n$ and the largest item $\wmax$ is well studied. Recent algorithms run in time $\tOh(n + \wmax^{1.5})$~\cite{Jin23,ChenLMZ23}, and a conditional lower bound rules out time $\wmax^{1-\delta} \cdot 2^{o(n)}$ for any $\delta > 0$~\cite{AbboudBHS22}, which leaves a gap between $\wmax$ and $\wmax^{1.5}$. Can the ideas in this paper help to close this gap?
\end{itemize}


\bibliography{main}

\appendix

\section{More Details for the Proof of Lemma~\ref{lem:constronestepguarantee}}
\label{sec:appproofproximity}

In the proof of Lemma~\ref{lem:constronestepguarantee} we argued Case 2.2 away by claiming that it is symmetric to Case 2.1. Here we spell out the details of the proof in that case, for completeness.

\begin{proof}[Proof of Lemma~\ref{lem:constronestepguarantee}]
  See Section~\ref{sec:proximity} for the majority of the proof. Here we provide the full proof of Case 2.2, for which previously we only argued that it is symmetric to Case 2.1.
  Recall that the assumption of Case 2 is $|I| > 6000 (\log^3(2n) \cdot m \, \wmax)^{1/2}$.
  
  \medskip
  \emph{Case 2.2: $I = I^-$.} We further split this case into two subcases.
  
  \medskip
  \emph{Case 2.2.1: $\sum_{i \in J^- \setminus I^-} |g_i - x^*_i| \le |J^- \setminus I^-|/2$.} 
  Let $I_X = \{ i \in J^- \setminus I^- : x^*_i = 1 \}$ and consider its multi-set of weights $X = w(I_X)$. 
  We claim that $X$ satisfies the assumption of Lemma~\ref{lem:addcomb}, i.e., $|X| \ge 1500 (\log^3(2 |X|) \cdot \mu(X) \, \wmax)^{1/2}$. We prove this claim in the remainder of this paragraph. Note that $|X| = |I_X| \ge |J^- \setminus I^-|/2$, since all items $i \in J^- \setminus I^-$ have $g_i = 1$ and thus by the assumption of Case 2.2.1 at least half of these items have $x^*_i = 1$. Since $|I^-| = \lceil |J^-|/2 \rceil$ we have $|J^- \setminus I^-| \ge |I^-|-1 = |I|-1 \ge |I|/2$, where we used $|I| > 1$ (by Case 2). Together with the assumption of Case 2 we obtain $|X| \ge |I|/4 \ge 1500 (\log^3(2n) \cdot m \, \wmax)^{1/2} \ge 1500 (\log^3(2 |X|) \cdot \mu(X) \, \wmax)^{1/2}$, since $|X| \le n$ and $\mu(X) \le n$. Thus, $X$ satisfies the assumption of Lemma~\ref{lem:addcomb}.
  
  Let $I_Y := \{ i \in I^- : x^*_i = 0 \}$ and consider its multi-set of weights $Y = w(I_Y)$. 
  
  Suppose that $\Sigma(Y)$ satisfies the assumption of Lemma~\ref{lem:addcomb}. Then we can apply Lemma~\ref{lem:addcomb} to $X$ and $Y$, which yields $X' \subseteq X$ and $Y' \subseteq Y$ with equal sum $\Sigma(X') = \Sigma(Y')$. 
  The corresponding subsets $I_{X'} \subseteq I_X$ and $I_{Y'} \subseteq I_Y$ have equal total weight $\Sigma(w(I_{X'})) = \Sigma(w(I_{Y'}))$. Since $\max(I_Y) < \min(I_X)$, this contradicts Lemma~\ref{lem:noequalsubsets}. 
  
  Therefore, $\Sigma(Y)$ cannot satisfy the assumption of Lemma~\ref{lem:addcomb}, i.e., we have 
  $$ \Sigma(Y) \le 340000 \log(2|X|) \mu(X) \wmax^2 / |X|. $$
  Now we again use $|I|/4 \le |X| \le n$ and $\mu(X) \le m$ to obtain
  $$ \Sigma(Y) \le 1360000 \log(2n) \cdot m \, \wmax^2 / |I| \le \Delta. $$
  Note that since $I = I^-$, all items $i \in I$ have $g_i = 1$ and thus $\Sigma(Y) = \Sigma(w(I_Y)) = \sum_{i \in I} w_i |g_i - x^*_i|$.
  We thus obtain $\sum_{i \in I} w_i |g_i - x^*_i| \le \Delta$, as desired. 
  
  \medskip
  \emph{Case 2.2.2: $\sum_{i \in J^- \setminus I^-} |g_i - x^*_i| > |J^- \setminus I^-|/2$.} 
  Let $I_X = \{ i \in J^- \setminus I^- : x^*_i = 0 \}$ and consider its multi-set of weights $X = w(I_X)$. 
  We claim that $X$ satisfies the assumption of Lemma~\ref{lem:addcomb}, i.e., $|X| \ge 1500 (\log^3(2 |X|) \cdot \mu(X) \, \wmax)^{1/2}$. We prove this claim in the remainder of this paragraph. Note that $|X| = |I_X| \ge |J^- \setminus I^-|/2$, since all items $i \in J^- \setminus I^-$ have $g_i = 1$ and thus by the assumption of Case 2.2.2 at least half of these items have $x^*_i = 0$. 
  The rest of the argument is exactly as in the previous case, and we arrive at $|X| \ge 1500 (\log^3(2 |X|) \cdot \mu(X) \, \wmax)^{1/2}$, so $X$ satisfies the assumption of Lemma~\ref{lem:addcomb}.
  
  Let $I_Y := \{ i \in [n] : g_i = 0, x^*_i = 1 \}$ and consider its multi-set of weights $Y = w(I_Y)$. 
  
  Suppose that $\Sigma(Y)$ satisfies the assumption of Lemma~\ref{lem:addcomb}. Then we can apply Lemma~\ref{lem:addcomb} to $X$ and $Y$, which yields $X' \subseteq X$ and $Y' \subseteq Y$ with equal sum $\Sigma(X') = \Sigma(Y')$. 
  The corresponding subsets $I_{X'} \subseteq I_X$ and $I_{Y'} \subseteq I_Y$ have equal total weight $\Sigma(w(I_{X'})) = \Sigma(w(I_{Y'}))$. Since $\max(I_X) < \min(I_Y)$, this contradicts Lemma~\ref{lem:noequalsubsets}. 
  
  Therefore, $\Sigma(Y)$ cannot satisfy the assumption of Lemma~\ref{lem:addcomb}, i.e., we have 
  $$ \Sigma(Y) \le 340000 \log(2|X|) \mu(X) \wmax^2 / |X|. $$
  Now we again use $|I|/4 \le |X| \le n$ and $\mu(X) \le m$ to obtain
  $$ \Sigma(Y) \le 1360000 \log(2n) \cdot m \, \wmax^2 / |I| \le \Delta. $$
  It remains to relate $\Sigma(Y)$ to $\sum_{i \in I} w_i |g_i - x^*_i|$. 
  Note that $x^*$ is maximal in the sense that no item can be added to it, and $g$ is maximal in the sense that it selects a maximal prefix of the items. It follows that either $w^T x^* = w^T g = \sum_{i \in [n]} w_i$ or
  $w^T x^*, w^T g \in (W-\wmax, W]$. Both yield 
  $|w^T x^* - w^T g| < \wmax$. Moreover, we have
  $$ w^T x^* - w^T g = \sum_{i \in [n], g_i = 0} w_i |g_i - x^*_i| - \sum_{i \in [n], g_i = 1} w_i |g_i - x^*_i|,  $$
  and thus
  $$ \sum_{i \in I} w_i |g_i - x^*_i| \le \sum_{i \in [n], g_i = 1} w_i |g_i - x^*_i| < \sum_{i \in [n], g_i = 0} w_i |g_i - x^*_i| + \wmax. $$
  Finally, we note that
  $$ \Sigma(Y) = \Sigma(w(I_Y)) = \sum_{i \in [n], g_i = 0} w_i |g_i - x^*_i|, $$
  which implies the desired
  $\sum_{i \in I} w_i |g_i - x^*_i| \le \Sigma(Y) + \wmax \le \Delta$.
\end{proof}

\end{document}